\documentclass[12pt]{article}
\usepackage[utf8]{inputenc}
\usepackage[cm]{fullpage}
\usepackage{amsmath}
\usepackage{amsthm}
\usepackage{amsfonts}
\usepackage{amssymb}
\usepackage{esint}
\usepackage{color}
\numberwithin{equation}{section}
\newtheorem{theorem}{Theorem}
\newtheorem{definition}{Definition}[subsection]
\newtheorem{lemma}{Lemma}[subsection]
\begin{document}

\title{Three-Body Dispersive Scattering}
\author{Michael Breeling and Avy Soffer}
\maketitle

\begin {abstract}
We study the spectral and scattering theory of three body dispersive systems,
which include a massless particle and a two body non-relativistic pair, along with two body short
interactions among the three particles.
We prove local decay estimates and propagation estimates, which are then used to prove asymptotic completeness for all non-threshold negative energies of the system. By adding (trace class) non-particle-conserving interactions, this models ionization processes as well. In particular, if the massless particle is interacting only through the non-particle-number-conserving interactions, no ionization is possible if the total energy of the system is negative, as predicted by the formal analysis of the photo-electric effect.
\end{abstract}
\tableofcontents
\section{Introduction}

Dispersive equations include N-quasiparticle systems of photons, phonons, spin waves and magnons, Bogoliubov-De Gennes quasiparticles and more.\bigskip

Mathematical theory of N-body scattering is based on the special dispersion relation of quantum particles. When the particles are not standard nonrelativistic, such as photons,
or other dispersion-based particles, the theory is largely open.
See however the works \cite{gerard},\cite{zie1}. \bigskip

In the context of quantum field theory, there has been recent progress, in particular by the works of Sigal and collaborators \cite{sf}.\bigskip

In this work we've been studying the special case of three particles, one of which is a massless particle, a ``photon''. In particular, we prove asymptotic completeness and propagation estimates for such systems.\bigskip

Since the kinetic part is not quadratic, the standard proofs have to be modified. Furthermore, the singular nature of the dispersion relation of the ``photon'', requires a new approach to micro- localization of the phase-space operators.\bigskip

In these systems there are three new mathematical problems: for one, since the kinetic part is not quadratic, it is in general not known how to prove local decay estimates using positive commutator methods. The standard conjugate operator given by the dilation generator does not have a positive commutator with the kinetic part. On the other hand, the dispersive generalization of the dilation, based on the kinetic part, results in a non-local operator with an unbounded and uncontrolled commutator with N-body potential terms.\bigskip

The second problem has to do with the fact that some dispersion relations are singular, like that of a photon. It is singular at 0 frequency.\bigskip

The third problem is that typically the number of particles is not conserved by the Hamiltonian, and in cases when the quasiparticle is massless, arbitrary number of such particles can emerge during the interaction.\bigskip

At the more technical level the bound states of subcluster Hamiltonians depend on the external momentum. Hence the typical bound cluster is in fact given by a fiber integral over all possible external momenta of the cluster.\bigskip

In this work, we study the spectral and scattering theory of a three-body dispersive system, consisting of a massless particle as well, interacting with non-relativistic particles via general short range two-body interactions.\bigskip

We allow arbitrary numbers of bound states for the subsystems, and ionization process is also allowed.\bigskip

We prove, via a Mourre estimate, propagation estimates and local decay. These are then used to prove asymptotic completeness for all negative energies, away from thresholds. The case of positive energies, just like in the standard three-body scattering, is more complicated and will be done separately.\bigskip

The condition of negative energy on the initial state does not preclude ionization, since the outgoing state in our case can have a free ``electron'' and total energy negative (coming from a bound state of the massless particle with either the ``electron'' or ``proton'').\bigskip

If on the other hand, we restrict the interaction between the massless particle and the rest of the system to be a trace class perturbation which involves the creation/annihilation of the massless particle, then, by energy conservation, ionization will take place only if the incoming massless particle is energetic enough to free the ``electron''.\bigskip

To obtain these results we prove propagation estimates which are adapted to the massless case. Special care is needed in the development of commutator estimates, due to the singular nature of the dispersion relation:

$$
H_0 = p^2+|k|,
$$

where the first term corresponds to the ``electron'' with mass 2, the second term is the massless particle (with momentum operator $k$), and the third particle is assumed to have infinite mass, and located at the origin.\bigskip

So, the full Hamiltonian is of the form

$$
H_0+V_{12}(x)+V_{13}(y)+V_{23}(x-y)
$$

where $x$ is the position of the ``electron'', $y$ the position of the massless particle and $V_{23}$ is the interaction between the massless particle and the ``electron''.\bigskip

The space dimension is 3.\bigskip

Other complications arise due to the dispersive (non-quadratic) nature of $H_0$.\bigskip

In this case the bound states (energies and eigenfunctions) of the subsystems depend explicitly on the momentum of the moving subsystem. That is, a moving electron with a photon cloud has a cloud which depends on its velocity (as the magnetic field around a conducting wire depends on the current).\bigskip

So, in fact, the asymptotic states are represented as fiber integrals over the external momentum, of momentum dependent eigenfunctions in the internal coordinate,  $x-y$, for the ``electron'' + massless particle cluster.

\section{Definitions and potential assumptions}

Let $\mathbb{R}^6_X=\mathbb{R}^3_x\oplus \mathbb{R}^3_y$ be the configuration space for a restricted three-particle system, where coordinates have been selected so that a first particle (a ``proton'', treated as having infinite mass) is at the origin, $x$ is the position of the second particle (an ``electron''), and $y$ is the position of the third particle (a ``photon''). Where $X=(x,y)$ represents a point in configuration space, we let $P=(p,k)$ be the Fourier conjugate of $X$, so that $k$ is the momentum of the photon and $p$ is the momentum of the electron. We define the \textbf{free Hamiltonian}
$$H_0 := p^2+|k|$$
as an operator on $L^2(\mathbb{R}^6_P)$. The free Hamiltonian models the kinetic energies $p^2$ and $|k|$ of the electron and of a photon, respectively. The proton is fixed at the origin and so makes no contribution, and creation and annihilation of photons is ignored in this simplified model. As a multiplier operator, $H_0$ is self-adjoint on the domain $D(H_0) := \lbrace \psi \in L^2(\mathbb{R}^6) : H_0\psi \in L^2(\mathbb{R}^6)\rbrace$. The graph norm on $D(H_0)$ is equivalent to the weighted $L^2$ norm with weight  $(1+p^4+k^2)\text{ } dp\text{ } dk$.   As such, the operator $H_0$ is essentially self-adjoint on the dense set $C^\infty_0(\mathbb{R}^6)$, the class of smooth functions with compact support; this class is dense in the weighted space $L^2(\mathbb{R}^6_P; (1+p^4+k^2)\text{ } dp\text{ } dk)$. \bigskip

We introduce a three-body potential function $V = V_{12}(x)+V_{13}(y)+V_{23}(x-y)$, where $V_{12}$, $V_{13}$, and $V_{23}$ are functions on $\mathbb{R}^3$. $V_{12}$ is the electron-proton interaction, $V_{13}$ is the photon-proton interaction, and $V_{23}$ is the electron-photon interaction. We define the \textbf{full Hamiltonian} on $L^2(\mathbb{R}^6_X)$:

$$H= H_0 + V_{12}(x)+ V_{13}(y) + V_{23}(x-y)$$

We assume that the potentials satisfy the assumptions of the Kato-Rellich theorem so that $H$ is self-adjoint on $D(H)=D(H_0)$ and bounded below. In particular, we make assumption (\ref{RB}) below. \bigskip

Next we define the five \textbf{cluster decompositions} $a \in \lbrace (xy0), (y)(x0), (x)(y0), (xy)(0), (x)(y)(0)\rbrace$. These are used as indices for quantities representing the system in the following situations respectively: when all three particles are close together, when just the photon is far away, when just the electron is far away, when the photon and electron are close together but far away from the proton, and when all three particles are far apart. The number of clusters for a particular decomposition is denoted $\#(a)$, e.g. $\#((x)(y0))=2$.\bigskip

When we are analyzing a particular cluster decomposition, it's useful to have a uniform notation to describe the internal and external cluster coordinates. We write $x^a$ and $p^a$ for internal coordinates and $x_a$ and $p_a$ for external coordinates, e.g. for $a=(y)(x0)$ we write $y=x_a$, $x=x^a$, $k=p_a$, and $p=p^a$. The rest of the notation is contained in this chart.

\begin{center}
    \begin{tabular}{| l | l | l | l | l | l |}
    \hline
    & (xy0)& (y)(x0) & (x)(y0) & (xy)(0) & (x)(y)(0) \\ \hline
    $x_a$ & - & $y$ & $x$ & $x+y$ & $X$ \\ \hline
    $x^a$ & $X$ & $x$ & $y$ & $x-y$ & - \\ \hline
    $p_a$ & - & $k$ & $p$ & $p+k$ & $P$ \\ \hline
    $p^a$ & $P$ & $p$ & $k$ & $p-k$ & - \\ \hline
    $I^a(x^a)$ & $V$ & $V_{12}(x)$ & $ V_{13}(y)$ & $V_{23}(x-y)$ & $0$ \\ \hline
    $I_a$ & $0$ & $V_{13}(y)+V_{23}(x-y)$ & $V_{12}(x)+V_{23}(x-y)$ & $V_{12}(x)+V_{13}(y)$ & $V$ \\ \hline
    \end{tabular}
\end{center}

The assumption (\ref{RB}) is that the following operators are are relatively $H_0$-bounded with relative bound less than $1$: for all partitions $a$,

\begin{equation}\label{RB}
\begin{cases}
I^a\\
x^a \cdot \bigtriangledown I^a \\
x^a \cdot \bigtriangledown(x^a \cdot \bigtriangledown I^a)\\
\end{cases}\tag{RB}
\end{equation}

One simple way to satisfy (\ref{RB}) is to have each potential be smooth and decay at infinity, although weaker conditions will suffice.\bigskip

We define the \textbf{truncated Hamiltonians} $H_a=H_0+I^a = H-I_a$. These represent the approximate energy of the system when it has separated into clusters in the manner suggested by $a$, so we can neglect certain potentials. The $I_a$ are the \textbf{intercluster potentials}, and the $I^a$ are the internal potentials.\bigskip

We focus our attention on the three \textbf{2-cluster decompositions} such that $\#(a)=2$. For each of these cluster decompositions, the truncated Hamiltonians $H_a$ can be written as a direct integral (cf. \cite{gerardlaba}) over the \textbf{reduced Hamiltonians} $H_a(s)$ defined below. This is justified because $H_a$ commutes with $p_a$ in each case. Essentially, we get to identify $p_a$ with a number $s$. \bigskip

$H_{(y)(x0)}(s):= p^2 + |s| + V_{12}(x) $ is an operator on $L^2(\mathbb{R}^3_x)$, so that $H_{(y)(x0)} = \int^\oplus_{\mathbb{R}^3_s} H_{(y)(x0)}(s)\text{ } ds$ is an operator on $L^2(\mathbb{R}^6) = \int^\oplus_{\mathbb{R}^3_s} L^2(\mathbb{R}^3_x) \text{ } ds$.\bigskip

$H_{(x)(y0)}(s)=  s^2 + |k| + V_{13}(y)$ is an operator on $L^2(\mathbb{R}^3_y)$, so that $H_{(x)(y0)} = \int^\oplus_{\mathbb{R}^3_s} H_{(x)(y0)}(s)\text{ } ds$ is an operator on $L^2(\mathbb{R}^6) = \int^\oplus_{\mathbb{R}^3_s} L^2(\mathbb{R}^3_y) \text{ } ds$.\bigskip

$H_{(xy)(0)}(s)= \frac{1}{4}(p^a+s)^2 + \frac{1}{2}|p^a-s| + V_{23}(x^a)$ is an operator on $L^2(\mathbb{R}^3_{x^a})$ so that $H_{(xy)(0)} = \int^\oplus_{\mathbb{R}^3_s} H_{(xy)(0)}(s)\text{ } ds$ is an operator on $L^2(\mathbb{R}^6) = \int^\oplus_{\mathbb{R}^3_s} L^2(\mathbb{R}^3_{x^a}) \text{ } ds$.\bigskip

In order to prove the Mourre estimate, we use an additional spectral assumption on the reduced Hamiltonians $H_{(xy)(0)}(s)$ for the photon-electron cluster. Specifically, we assume
\begin{equation}\label{SPEC}
\begin{split}
&H_{(xy)(0)}(s) \text{ has no eigenvalues embedded in its continuous spectrum for any }s\in\mathbb{R}^3\\
\end{split}\tag{SPEC}
\end{equation}

This is because eigenvalues can behave poorly as $s$ varies, where the continuous spectrum is concerned. Isolated eigenvalues, on the other hand, have a well-understood structure (cf. \cite{kato}). It may be possible to remove the assumption of no singular continuous spectrum by proving an appropriate Mourre estimate.\bigskip

We define the \textbf{subsystem Hamiltonians} as

$$h_{(y)(x0)}= H_{(y)(x0)}(0)= p^2 + V_{12}(x)$$
$$h_{(x)(y0)}= H_{(x)(y0)}(0)= |k| + V_{13}(y)$$
$$h_{(xy)(0)}= H_{(xy)(0)}(0) = \frac{1}{4}(p^a)^2 + \frac{1}{2}|p^a| + V_{23}(x^a)$$

We further assume the following relative boundedness and compactness properties.

\begin{equation}\label{RC1}
\begin{cases}
V_{13}\text{ and }y\cdot \bigtriangledown V_{13}\text{ are relatively } |k|\text{-compact as operators on } L^2(\mathbb{R}^3_y) \\
V_{13}\text{ and }y\cdot \bigtriangledown V_{13}\text{ have relative } |k|\text{-bound less than 1} \\
V_{12}\text{ and }x\cdot \bigtriangledown V_{12}\text{ are relatively }p^2\text{-compact as operators on } L^2(\mathbb{R}^3_x)\\
V_{12}\text{ and }x\cdot \bigtriangledown V_{12}\text{ have relative }p^2\text{-bound bound less than 1} \\
V_{23}(x^a)\text{ and }x^a\cdot \bigtriangledown W (x^a)\text{ are relatively }(p^a+s)^2 \text{-compact }\\
\text{ as operators on }L^2(\mathbb{R}^3_{x^a}) \text{ for all } s\in \mathbb{R}^3\\
V_{23}(x^a)\text{ and }x^a\cdot \bigtriangledown W (x^a)\text{ have relative }(p^a+s)^2  \text{-bound }\\
\text{ less than 1 for all } s\in \mathbb{R}^3\\
\end{cases}\tag{RC1}
\end{equation}

These hold if, for instance, each potential function is continuous and decays at infinity. The Kato-Rellich theorem and (\ref{RC1}) imply that the reduced and subsystem Hamiltonians are self-adjoint on their respective domains $D(|k|)$, $D(p^2)$, and $D((p^a)^2+|p^a|)$.\bigskip

The eigenvalues of the subsystem Hamiltonians, along with zero, form the set of \textbf{thresholds}. The threshold energies are significant; if the entire system is given a threshold energy, then a subsystem may form a bound state without any kinetic energy left over to separate it from the remaining particle. This presents complications in sections below.\bigskip

\subsection{Definitions involving commutators}

Since we will frequently need to commute unbounded operators, but commutators of unbounded operators are a priori only defined as quadratic forms, the following is convenient. For self-adjoint operators $H$ and $A$, it is known that if $H$ satisfies the $C^1(A)$ property (cf. \cite{abg}), then the commutator $[H,iA]$ is well-defined and in fact the virial theorem holds. For more, see Proposition II.1 of \cite{mourre} and conditions $(M)$ and $(M')$ from \cite{virial}. \bigskip

\begin{lemma}[Formal commutators are well defined]\label{commlemma}
Suppose we have two (possibly unbounded) self-adjoint operators $H$ and $A$, where $H$ is bounded below. A priori there exists a quadratic form $[H,iA]_0$ on $D(H)\cap D(A)$. Suppose that $[H,iA]_0$ evaluated on a space $S$ of test vectors agrees (on $S$) with the closed quadratic form associated with a self-adjoint operator $C$ defined on an operator domain $D(C)$. Then, under the following conditions:

\begin{equation}\label{formal1}
\text{}e^{itA}\text{ maps both }D(H)\text{ and }S\text{ into themselves.}\tag{FC1}
\end{equation}
\begin{equation}\label{formal2}
S\subset D(H) \cap D(A)\tag{FC2}
\end{equation}
\begin{equation}\label{formal3}
S\text{ is a core for }D(H)\tag{FC3}
\end{equation}
\begin{equation}\label{formal4}
D(H) \subset D(C)\tag{FC4}
\end{equation}

We have that $H$ is $C^1(A)$, $[H,iA]_0$ is closeable, the self-adjoint operator associated to its closure is $C$, and the virial theorem

$$\langle \psi, C \psi \rangle = 0 \text{ whenever }\psi\text{ is an eigenvector of H.}$$

holds.
\end{lemma}

We always write $E_\bigtriangleup(H)$ for the spectral projection of the operator $H$ onto the interval $\bigtriangleup\subset \mathbb{R}$.\bigskip

\begin{definition}[Mourre estimate]
Let $H$ and $A$ be self-adjoint operators on a Hilbert space, let $E\in \mathbb{R}$, and let $\alpha>0$. Suppose that $H$ and $A$ satisfy the hypotheses of Lemma \ref{commlemma}, so that $[H,iA]$ is a well-defined self-adjoint operator with domain containing $D(H)$. $H$ is said to satisfy a Mourre estimate at $E$ with conjugate operator $A$, constant $\alpha$, and width $\delta$ if there exists an open interval $\bigtriangleup = (E-\delta,E+\delta)$  and a compact operator $K$ such that

$$E_\bigtriangleup(H) [H,iA] E_\bigtriangleup(H) \ge \alpha E_\bigtriangleup(H) + K$$

\end{definition}

We define the Mourre conjugate operators
$$A^a := \frac{1}{2}\left(p^a \cdot x^a + x^a \cdot p^a \right)$$
$$A_a := \frac{1}{2}\left(p_a \cdot x_a + x_a \cdot p_a \right)$$
$$A := A^a+A_a = \frac{1}{2}\left(P \cdot X + X \cdot P \right)$$
Since $C^\infty_0(\mathbb{R}^n)$ is invariant under dilations $\psi(r)\mapsto e^{-nt/2}\psi(e^t r)$ and these form a strongly continuous unitary group, all these operators (the generators of the unitary dilation groups) are known to be essentially self-adjoint on their respective $C^\infty_0$ spaces (cf. Theorem VIII.10 of \cite{rs}). \bigskip

Lemma \ref{commlemma} applies to our $H$ and $A$, with $S$ being the class of $C^\infty_0$ functions. Condition (\ref{formal1}) is satisfied since the $e^{itA}$ are dilations, which map both $L^2(\mathbb{R}^6; (1+k^2+p^4)\text{ }dk\text{ }dp)$ and $C^\infty_0(\mathbb{R}^6)$ into themselves. Condition (\ref{formal2}) is satisfied as well from elementary properties of smooth functions with compact support. Condition (\ref{formal3}) is satisfied by the Kato-Rellich theorem; since $S$ is a core for $H_0$, $S$ is also a core for $H$. To show condition (\ref{formal4}), we examine the formal commutator

$$C = 2p^2 + |k| - \sum_{\#(a)=2} x^a \cdot \bigtriangledown I^a $$

 From another application of Kato-Rellich, this has the same domain as $H$. Thus the lemma applies, and $[H,iA]$ is thought of as extending to the operator $C$. Similarly, we may compute the commutators:

$$[H_a,iA]= 2p^2 + |k| - x^a \cdot \bigtriangledown I^a \text{ for each } a$$
$$[h_{(y)(x0)},iA^a]= 2p^2 -x\cdot V_{12}(x)$$
$$[h_{(x)(y0)},iA^a]= |k| - y\cdot V_{13}(y)$$
$$[h_{(xy)(0)},iA^a] = \frac{1}{2}(p^a)^2 + \frac{1}{2}|p^a| -x^a \cdot V_{23}(x^a)$$
$$[H_{(xy)(0)}(s) , i A^a] = \frac{1}{2}(p^a)^2+ \frac{1}{2}p^a \cdot s +\frac{1}{2}\frac{(p^a)^2-s\cdot p^a}{|p^a-s|}-x^a \cdot V_{23}(x^a)$$

Since the constant involved in the Mourre estimate for our $H$ depends on the distance from thresholds, before stating it we must describe the spectra of the subsystem Hamiltonians. The subsystem Hamiltonians all satisfy a Mourre estimate at all nonzero energies $E$ (with conjugate operators $\text{sgn}(E) A^a$). The arguments are standard and use little more than the functional calculus, all following (\cite{mourre}, Theorem I.1). Because of this, we would like to use Mourre theory to conclude that eigenvalues of each $h_a$ may only accumulate at $0$, and each $h_a$ has no continuous singular  spectrum. Moreover, since the essential spectrum of each $h_a$ must be $[0,\infty)$ by Weyl's theorem (using (\ref{RC1})), there is no continuous spectrum below $0$. But to draw these conclusions from the result in \cite{mourre}, we need to verify also a condition on the second commutators. Self-adjoint operators $H$ and $A$ are said to satisfy condition (\ref{2COMM}) if

\begin{equation}\label{2COMM}
\begin{split}
& H \text{ and } A \text{ satisfy the hypotheses of Lemma \ref{commlemma}, and given the operator } C \text{ extending }\\
& [H,iA] \text{, we have that the hypotheses are also satisfied by } C \text{ and } A \text{, so that } [C,iA]\\
& \text{ extends to a a self-adjoint operator with domain containing } D(H) \text{.}\\
 \end{split}\tag{2COMM}
\end{equation}

Computing the second commutators

$$[[h_{(y)(x0)},iA^a],iA^a]= 4p^2 +x\cdot \bigtriangledown(x\cdot \bigtriangledown V_{12}(x))$$
$$[[h_{(x)(y0)},iA^a],iA_y]= |k| + y \cdot \bigtriangledown (y\cdot \bigtriangledown V_{13}(y))$$
$$[[h_{(xy)(0)},iA^a],iA^a] = (p^a)^2 + \frac{1}{2}|p^a| +x^a \cdot \bigtriangledown ( x^a \cdot \bigtriangledown V_{23}(x^a))$$

we confirm that condition (\ref{2COMM}) is satisfied in each case if we assume (\ref{RB2}):

\begin{equation}\label{RB2}
\begin{cases}
 y \cdot \bigtriangledown (y\cdot \bigtriangledown V_{13}(y))\text{ has relative } |k|\text{-bound less than 1} \\
x\cdot \bigtriangledown(x\cdot \bigtriangledown V_{12}(x))\text{ has relative }p^2\text{-bound bound less than 1} \\
x^a \cdot \bigtriangledown ( x^a\cdot \bigtriangledown V_{23}(x^a))\text{ has relative } (p^a)^2 + \frac{1}{2}|p^a| \text{-bound }\\
\text{ less than 1}\\
\end{cases}\tag{RB2}
\end{equation}

and apply Kato-Rellich. Thus the conditions of Mourre's paper \cite{mourre} are satisfied, and we may invoke the result that  eigenvalues of each $h_a$ may only accumulate at $0$, and each $h_a$ has no continuous singular spectrum (the other part about Weyl's theorem holds irregardless). We can now use our understanding of the spectra.\bigskip

For each of the 2-cluster decompositions $a$ define $G_a = \inf \lbrace \lbrace \lambda : \lambda \text{ eigenvalue of } h_a \rbrace \cup \lbrace 0 \rbrace \rbrace$; this is either an eigenvalue of $h_a$ or $0$, because eigenvalues of $h_a$ may only accumulate at $0$. Then define $d(E,a)$ to be

$$d(E,a) = \begin{cases}
\left( E - \sup \lbrace \lambda: \lambda \text{ zero or eigenvalue of } h_a \text{ , } \lambda <E \rbrace \right)  :  \text{ E not zero nor eigenvalue of } h_a \text{ , } E > G_a \\
0  :  \text{ E zero or eigenvalue of } h_a \\
b : \text{ } E < G_a\\
\end{cases}$$

where $b$ can be any positive constant one wishes. Roughly speaking, $d(E,a)$ is the distance from $E$ to the nearest eigenvalue of $h_a$ to the left of $E$. Then, let $d(E) = \min_a d(E,a)$. Then $d(E)$ is,  roughly, the distance from $E$ to the nearest threshold of any type to the left of $E$. \bigskip

Finally we are in a position to state the first main theorem.

\begin{theorem}\label{mourre}
Suppose that the potential functions satisfy (\ref{RB}), (\ref{RB2}),  (\ref{RC1}), and (\ref{SPEC}), as well as (\ref{RC2}) below. Then for any $\epsilon>0$ and any nonthreshold energy $E$, the full Hamiltonian $H$ satisfies a Mourre estimate at energy $E$ with conjugate operator $A$ and constant $\alpha= d(E)-\epsilon$.
\end{theorem}

The second main theorem is the statement of asymptotic completeness. The unitary evolutions $e^{itH}$ and $e^{itH_a}$ are all well-defined by the functional calculus. Following \cite{ss}, we say the evolution defined by $H$ is \textbf{asymptotically complete} if it is asymptotically clustering at all nonthreshold energies. That is, if for every nonthreshold energy $E$, there exists an interval $\bigtriangleup$ containing $E$ so that whenever $\psi \in ran(E_\bigtriangleup(H))$ is orthogonal to the eigenfunctions of $H$, there exists $\lbrace \phi_a\rbrace_{\# (a)>1}$ such that

$$\lim_{t\rightarrow\infty} \Vert e^{-itH} \psi - \sum_{\# (a)>1} e^{-itH_a} \phi_a \Vert = 0$$

We prove asymptotic completeness under the additional assumptions of negative energy, short range, and exponential decay of eigenfunctions.

\begin{theorem}\label{completeness}
Suppose the potential functions satisfy (\ref{RB}), (\ref{RB2}), and (\ref{RC1}), as well as (\ref{RC2}), (\ref{SR}), and (\ref{FDE}) below. Then, for every \emph{negative} nonthreshold energy $E$, there exists an interval $\bigtriangleup$ containing $E$ so that whenever $\psi \in ran(E_\bigtriangleup(H))$ is orthogonal to the eigenfunctions of $H$, there exists $\lbrace \phi_a\rbrace_{\# (a)>1}$ such that
$$\lim_{t\rightarrow\infty} \Vert e^{-itH} \psi - \sum_{\# (a)>1} e^{-itH_a} \phi_a \Vert = 0$$
\end{theorem}

Due to the Kato-Rosenblum theorem, one can also add an arbitrary trace class perturbation to $H$ and retain the result. This should enable the study of models such as in \cite{jensen}.\bigskip

Section 3 is dedicated to the proof of Theorem \ref{mourre}. In section 4, the Mourre theory is used to prove local decay and minimal velocity estimates. In Section 5, these are used to prove Theorem \ref{completeness}.

\section{Proof of the Mourre estimate}

\subsection{Configuration space partition of unity}

We can break up the proof of the Mourre estimate for $H$ into problems involving each $H_a$ by deploying a configuration space partition of unity (e.g. \cite{froeseherbst}, but originally due to Deift and Simon). We need functions $\lbrace j_a(x,y)\rbrace_a$ on $\mathbb{R}^6$ satisfying the following requirements:

\begin{itemize}
\item $\sum_a j_a^2 = 1$ and each $j_a$ is $C^\infty$.
\item Each $j_a$ is homogeneous of degree $0$ outside of the unit ball; in particular, derivatives of any $j_a$ of any order are relatively $H_0$-compact since they decay in all directions.
\item The following multiplier operators are relatively $H_0$-compact (and they remain so if $j_a$ is replaced with any derivative of $j_a$).

\begin{equation}\label{RC2}
\begin{cases}
j_{(xy0)}\\
I_a j_a \\
[I_a, iA] j_{a}\\
\end{cases}\tag{RC2}
\end{equation}
\end{itemize}

In the case that the potentials are continuous and decaying in all directions, these relative compactness properties can be achieved by selecting the partition of unity so that all of the above functions decay in all directions in $\mathbb{R}^6$. The assumption (\ref{RC2}) on the potentials is that these operators are indeed relatively $H_0$-compact for the $j_a$ thus constructed. We proceed with the construction.\bigskip

It will suffice to construct $j_a$ that satisfy the following support conditions:\bigskip

For some positive constants $C_0,\dots,C_{11}$
\begin{itemize}
\item $j_{(x)(y0)}$ is supported in $\lbrace X=(x,y) \in \mathbb{R}^6 : |x| > C_0 |X|, |y| < C_1 |X|, |x-y| > C_2|X| \rbrace$.
\item $j_{(y)(x0)}$ is supported in $\lbrace X=(x,y) \in \mathbb{R}^6 : |y| > C_3 |X|, |x|< C_4 |X|, |x-y| > C_5|X| \rbrace$.
\item $j_{(xy)(0)}$ is supported in $\lbrace X = (x,y) \in \mathbb{R}^6 : |x| > C_6 |X|, |y| > C_7 |X|, |x-y|< C_8 |X| $.
\item $j_{(x)(y)(0)}$ is supported in $\lbrace X = (x,y) \in \mathbb{R}^6 : |x| > C_9 |X|, |y| > C_{10}|X|, |x-y| > C_{11} |X|\rbrace$
\end{itemize}

If these support conditions hold, then in the case that the potential functions are smooth and decay in all directions in $\mathbb{R}^3$, all the functions in (\ref{RC2}) indeed decay in all directions in $\mathbb{R}^6$. We now construct such functions $\lbrace j_a \rbrace_a$. Define the following sets on the unit sphere $\mathbb{S}^5 \in \mathbb{R}^6$:

$$U_{(x)(y0)}:=\lbrace X = (x,y) : |y| < \frac{1}{20}, |x| > \frac{1}{10} \rbrace$$
$$U_{(y)(x0)}:=\lbrace X = (x,y) : |x| < \frac{1}{20}, |y| > \frac{1}{10} \rbrace$$
$$U_{(xy)(0)}:=\lbrace X = (x,y) : |x| > \frac{1}{30}, |y| > \frac{1}{30},|x-y| < \frac{1}{10} \rbrace$$
$$U_{(x)(y)(0)}:=\lbrace X = (x,y): |x| > \frac{1}{30}, |y| > \frac{1}{30} |x-y| > \frac{1}{20} \rbrace$$

By the parallelogram law $2|X|^2=|x+y|^2+|x-y|^2$ and the Pythagorean theorem $X^2=|x|^2+|y|^2$, these sets are nonempty and form an open cover of the unit sphere. We are guaranteed the existence of a traditional partition of unity $\lbrace \chi_a\rbrace_a$ subordinate to the four sets in this open cover. By assuming homogeneity of degree $0$, we can extend these functions to a traditional partition of unity on $\mathbb{R}^6\setminus\lbrace 0 \rbrace$. \bigskip

We can also partition unity in $\mathbb{R}^6$ into two functions $\chi_0$ and $\chi_1$ with the following properties: $\chi_0$ is supported entirely within the unit sphere, and $\chi_1$ is supported entirely away from $0$ while being equal to $1$ outside the unit sphere.\bigskip

Now, we can define $\widetilde{j}_{(xy0)} := \chi_0$ and $\widetilde{j}_a := \chi_1 \chi_a$ for the other clusters $a$. Then, by declaring $j_a := \widetilde{j}_a/\sqrt{\sum_a (\widetilde{j}_a)^2}$ we get the $\lbrace j_a\rbrace_a$ with all the desired properties.

\subsection{Breaking apart the main estimate}

We claim that for any $f\in C^\infty_0(\mathbb{R})$,

\begin{equation}\label{ims}
\sum_a f(H)[H,iA]f(H)j_a^2 = \text{(compact operators)} + \sum_{\#(a)>1} j_a f(H_a)[H_a,iA]f(H_a)j_a
\end{equation}

Fix an $f\in C^\infty_0(\mathbb{R})$. The $a=(xy0)$ term is compact; the proof is reserved for the next section. Then, for each cluster decomposition $a\neq (xy0)$, we can commute around the associated term.
\begin{align*}
f(H)[H,iA]f(H)j_a^2 &= f(H)[H,iA]f(H)j_a^2 - f(H)[H,iA]f(H_a)j_a^2\\
&\text{\hspace{.5cm}}+ f(H)[H,iA]f(H_a)j_a^2 - f(H)[H,iA]j_af(H_a)j_a\\
&\text{\hspace{.5cm}}+ f(H)[H,iA]j_af(H_a)j_a - f(H)j_a[H,iA]f(H_a)j_a\\
&\text{\hspace{.5cm}}+ f(H)j_a[H,iA]f(H_a)j_a -f(H_a)j_a[H,iA]f(H_a)j_a\\
&\text{\hspace{.5cm}}+ f(H_a)j_a[I_a,iA]f(H_a)j_a + f(H_a)j_a[H_a,iA]f(H_a)j_a\\
&\text{\hspace{.5cm}}- j_a f(H_a)[H_a,iA]f(H_a)j_a + j_a f(H_a)[H_a,iA]f(H_a)j_a \\
&= f(H)[H,iA]\left(f(H)-f(H_a)\right)j_a^2 + f(H)[H,iA][f(H_a),j_a]j_a \\
&\text{\hspace{.5cm}}+ f(H)[[H,iA],j_a]f(H_a)j_a + \left(f(H)-f(H_a)\right)j_a[H,iA]f(H_a)j_a\\
&\text{\hspace{.5cm}}+ f(H_a)j_a[I_a,iA]f(H_a)j_a + [f(H_a),j_a][H_a,iA]f(H_a)j_a\\
&\text{\hspace{.5cm}}+ j_af(H_a)[H_a,iA]f(H_a)j_a\\
\end{align*}

This is an equality of bounded operators. The operator $[I_a,iA]$ should be understood not as a closure of a formal commutator $[I_a,iA]_0$ in its own right, but rather as $[H,iA]-[H_a,iA]$ which is a self-adjoint operator having domain $D(H)$. \bigskip

It remains to show that all of these terms except for $j_af(H_a)[H_a,iA]f(H_a)j_a$ are compact.

\subsection{Compactness}

We repeatedly apply the following basic tool (indeed this was already used to describe some compactness properties of the partition of unity).

\begin{lemma}[Elementary compactness lemma]\label{elemcomp}
Suppose that $f(x)$ and $g(p)$ are continuous functions $\mathbb{R}^n\rightarrow \mathbb{C}$ such that $f(x)$ and $g(p)$ decay at infinity. Then, $f(x)g(p)$ is a compact operator on $L^2(R^n)$.
\end{lemma}

The proof is omitted.\bigskip

At this point, we fix a cluster decomposition $a\neq(xy0)$. It is helpful to isolate the following.

\begin{lemma}[Some compactness]\label{somecompact}
The following operators are compact for all decompositions $a$.
\begin{equation}\label{compact2}
[f(H_a),j_a]
\end{equation}
\begin{equation} \label{compact1}
\left(f(H)-f(H_a)\right)j_a
\end{equation}
\begin{equation}\label{compact3}
f(H)[[H,iA],j_a]f(H_a)
\end{equation}
\begin{equation}\label{compact4}
j_a[I_a,iA]f(H_a)
\end{equation}
\end{lemma}

Note that this completes the proof of (\ref{ims}). To prove these operators are compact, it is convenient to replace the functions $f$ with resolvents. To this end we prove:

\begin{lemma}[Auxiliary compactness]\label{auxcompact}
The following operators are compact for all decompositions $a$.
\begin{equation}\label{compact6}
[\frac{1}{H_a+i},j_a]
\end{equation}
\begin{equation} \label{compact5}
\left(\frac{1}{H+i}-\frac{1}{H_a+i}\right)j_a
\end{equation}
\begin{equation}\label{compact7}
\frac{1}{H+i}[[H,iA],j_a]\frac{1}{H_a+i}
\end{equation}
\begin{equation}\label{compact8}
j_a[I_a,iA]\frac{1}{H_a+i}
\end{equation}
\end{lemma}

We start with the compactness of (\ref{compact6}). We compute
\begin{align*}
\frac{1}{H_a+i}j_a-j_a\frac{1}{H_a+i} &= \frac{1}{H_a+i}j_a(H_a+i)\frac{1}{H_a+i}-\frac{1}{H_a+i}(H_a+i)j_a\frac{1}{H_a+i} \\
&= \frac{1}{H_a+i}\left( j_a H_0 - H_0 j_a\right) \frac{1}{H_a+i}\\
\end{align*}
Since $\frac{1}{H_a+i}$ has range in $D(H)$ and $j_a$ maps $D(H)$ into itself, the above are equalities of bounded operators. Since

\begin{align*}
 \frac{1}{H_a+i}\left( j_a H_0 - H_0 j_a\right) \frac{1}{H_a+i} &= \frac{1}{H_a+i} (H_0 + i) \frac{1}{H_0+i}\left( j_a H_0 - H_0 j_a\right) \frac{1}{H_0+i}(H_0 + i) \frac{1}{H_a+i} \\
\end{align*}

and $(H_0 + i) \frac{1}{H_a+i}$ is bounded by the potential assumptions, we just need to show that the operator $\frac{1}{H_0+i}\left( j_a H_0 - H_0 j_a\right) \frac{1}{H_0+i}$ is compact. We obtain the following:

\begin{align*}
\frac{1}{H_0+i}\left( j_a H_0 - H_0 j_a\right) \frac{1}{H_0+i} &= \frac{1}{H_0+i}2(\bigtriangledown_x j_a) \cdot \bigtriangledown_x\frac{1}{H_0+i} + \frac{1}{H_0+i}(\bigtriangleup_x j_a)\frac{1}{H_0+i}\\
&\text{\hspace{.5cm}}+\frac{1}{H_0+i} \left( j_a |k| - |k| j_a \right)\frac{1}{H_0+i}\\
\end{align*}

Since $\frac{1}{H_0+i} (\bigtriangleup_x j_a)$ is compact by the fact that $\bigtriangleup_x j_a$ decays in all directions, the term $\frac{1}{H_0+i}(\bigtriangleup_x j_a)\frac{1}{H_0+i}$ is compact. Similarly, we have that $\frac{1}{H_0+i}2(\frac{d}{dx_\ell} j_a)$ is compact for $\ell\in \lbrace 1,2,3 \rbrace$. Moreover, $\frac{d}{dx_\ell}\frac{1}{H_0+i}$ is bounded for $\ell\in \lbrace 1,2,3 \rbrace$. Thus the term $\frac{1}{H_0+i}2(\bigtriangledown_x j_a) \cdot \bigtriangledown_x\frac{1}{H_0+i}$ is compact. We focus on the only remaining term, $\frac{1}{H_0+i} \left( j_a |k| - |k| j_a \right)\frac{1}{H_0+i}$, which takes more effort. We will use the `square root lemma' $|k| = \frac{1}{\pi}\int_0^\infty \frac{s^{-\frac{1}{2}}}{s+k^2} k^2 ds$.\bigskip

We assert that for all Schwartz functions $\psi(q,r)$:

\begin{equation}\label{schwartz0}
\frac{1}{H_0+i} \left( j_a |k| - |k| j_a \right)\frac{1}{H_0+i} \psi =  \frac{1}{\pi}\int_0^\infty \frac{1}{H_0+i} \left( j_a \frac{s^{-\frac{1}{2}}}{s+k^2}k^2 - \frac{s^{-\frac{1}{2}}}{s+k^2}k^2 j_a \right)\frac{1}{H_0+i} \psi\text{ } ds
\end{equation}

We need only prove that for $\psi$ Schwartz class, $\frac{1}{\pi}\int_0^\infty \frac{s^{-\frac{1}{2}}}{s+k^2}k^2 \psi\text{ } ds$ as a $L^2(\mathbb{R}^6)$-valued integral converges to $|k|\psi$. It evidently converges pointwise to $|k|\psi$; one can check that it is convergent as a Bochner integral.\bigskip

For Schwartz $\psi$, we may apply $j_a$, $\frac{s^{-1/2}}{s+k^2}$, and $k^2$ in any order without encountering domain issues, since each operator preserves the Schwartz class. Therefore on the Schwartz class the following hold:

\begin{equation}\label{schwartz1}
[j_a,k^2]=(\bigtriangleup_y j_a) + 2 (\bigtriangledown_y j_a) \cdot \bigtriangledown_y
\end{equation}

\begin{equation}\label{schwartz2}
[j_a,\frac{s^{-\frac{1}{2}}}{s+k^2}k^2] =  [j_a,\frac{s^{-\frac{1}{2}}}{s+k^2}]k^2 + \frac{s^{-\frac{1}{2}}}{s+k^2}[j_a,k^2]
\end{equation}

$$[j_a,\frac{s^{-\frac{1}{2}}}{s+k^2}] = \frac{s^{-\frac{1}{2}}}{s+k^2} \frac{s+k^2}{s^{-\frac{1}{2}}} j_a \frac{s^{-\frac{1}{2}}}{s+k^2}- \frac{s^{-\frac{1}{2}}}{s+k^2} j_a  \frac{s+k^2}{s^{-\frac{1}{2}}} \frac{s^{-\frac{1}{2}}}{s+k^2}$$

\begin{equation}\label{schwartz3}
 = -\frac{s^{-\frac{1}{2}}}{s+k^2} (s^{\frac{1}{2}}) ((\bigtriangleup_y j_a) + 2 (\bigtriangledown_y j_a) \cdot \bigtriangledown_y)\frac{s^{-\frac{1}{2}}}{s+k^2}
\end{equation}

Therefore, combining (\ref{schwartz1})-(\ref{schwartz3}) we obtain:

\begin{equation}\label{schwartz4}
[j_a,\frac{s^{-\frac{1}{2}}}{s+k^2}k^2] = -\frac{s^{-\frac{1}{2}}}{s+k^2} (s^{\frac{1}{2}}) ((\bigtriangleup_y j_a) + 2 (\bigtriangledown_y j_a) \cdot \bigtriangledown_y)\frac{s^{-\frac{1}{2}}}{s+k^2}k^2 + \frac{s^{-\frac{1}{2}}}{s+k^2}((\bigtriangleup_y j_a) + 2 (\bigtriangledown_y j_a) \cdot \bigtriangledown_y)
\end{equation}

Moreover,
\begin{align*}
 (2(\bigtriangledown_y j_a) \cdot \bigtriangledown_y ) &=  2 [(\bigtriangledown_y j_a) ,\bigtriangledown_y] +  2 (\bigtriangledown_y \cdot (\bigtriangledown_y j_a))\\
&= - 2 (\bigtriangleup_y j_a) + 2 (ik \cdot (\bigtriangledown_y j_a))\\
\end{align*}

so that by (\ref{schwartz4})

\begin{equation}\label{schwartz5}
[j_a,\frac{s^{-\frac{1}{2}}}{s+k^2}k^2] = \frac{s^{-\frac{1}{2}}}{s+k^2} (s^{\frac{1}{2}}) ( (\bigtriangleup_y j_a) - 2 (ik \cdot (\bigtriangledown_y j_a)))\frac{s^{-\frac{1}{2}}}{s+k^2}k^2 - \frac{s^{-\frac{1}{2}}}{s+k^2}( (\bigtriangleup_y j_a) - 2 (ik \cdot (\bigtriangledown_y j_a)))
\end{equation}

Using (\ref{schwartz5}), we can break up the integral in (\ref{schwartz0}).

$$\frac{1}{\pi}\int_0^\infty \frac{1}{H_0+i} \left( j_a \frac{s^{-\frac{1}{2}}}{s+k^2}k^2 - \frac{s^{-\frac{1}{2}}}{s+k^2}k^2 j_a \right)\frac{1}{H_0+i} \psi\text{ } ds$$
\begin{align*}
&= \frac{1}{\pi}\int_0^\infty \frac{1}{H_0+i} \left( - \frac{s^{-\frac{1}{2}}}{s+k^2} (\bigtriangleup_y j_a)\right)\frac{1}{H_0+i}\psi\text{ } ds \\
&\hspace{.5cm}+ \frac{1}{\pi}\int_0^\infty \frac{1}{H_0+i} \left( \frac{s^{-\frac{1}{2}}}{s+k^2} (2 ik \cdot (\bigtriangledown_y j_a))\right)\frac{1}{H_0+i}\psi\text{ } ds \\
&\hspace{.5cm}+ \frac{1}{\pi}\int_0^\infty \frac{1}{H_0+i} \left( \frac{s^{-\frac{1}{2}}}{s+k^2} (s^{\frac{1}{2}}) (\bigtriangleup_y j_a)\frac{s^{-\frac{1}{2}}}{s+k^2}k^2 \right)\frac{1}{H_0+i}\psi\text{ } ds \\
&\hspace{.5cm}+ \frac{1}{\pi}\int_0^\infty \frac{1}{H_0+i}  \left( -\frac{s^{-\frac{1}{2}}}{s+k^2} (s^{\frac{1}{2}}) (2 ik \cdot (\bigtriangledown_y j_a))\frac{s^{-\frac{1}{2}}}{s+k^2}k^2 \right)\frac{1}{H_0+i}\psi\text{ } ds \\
\end{align*}

Motivated by this, we compute the following operator norms for fixed $s$ and show that they are integrable functions of $s$ on $[0,\infty)$.

\begin{equation}\label{integrand1}\Vert \frac{1}{H_0+i} \frac{s^{-\frac{1}{2}}}{s+k^2} (\bigtriangleup_y j_a) \frac{1}{H_0+i}\Vert_{L^2\rightarrow L^2}
\end{equation}
\begin{equation}\label{integrand2}\Vert \frac{1}{H_0+i} \frac{s^{-\frac{1}{2}}}{s+k^2} (k \cdot (\bigtriangledown_y j_a)) \frac{1}{H_0+i} \Vert_{L^2\rightarrow L^2}
\end{equation}
\begin{equation}\label{integrand3}\Vert \frac{1}{H_0+i} \frac{s^{-\frac{1}{2}}}{s+k^2} (s^{\frac{1}{2}}) (\bigtriangleup_y j_a )\frac{s^{-\frac{1}{2}}}{s+k^2} (k^2) \frac{1}{H_0+i}\Vert_{L^2 \rightarrow L^2}
\end{equation}
\begin{equation}\label{integrand4}\Vert \frac{1}{H_0+i} \frac{s^{-\frac{1}{2}}}{s+k^2} (s^{\frac{1}{2}}) ( k \cdot (\bigtriangledown_y j_a))\frac{s^{-\frac{1}{2}}}{s+k^2}(k^2) \frac{1}{H_0+i} \Vert_{L^2 \rightarrow L^2}
\end{equation}
\bigskip

This would mean that the operator-valued integral

\begin{equation}\label{schwartz6}
\frac{1}{\pi}\int_0^\infty \frac{1}{H_0+i} \left( j_a \frac{s^{-\frac{1}{2}}}{s+k^2}k^2 - \frac{s^{-\frac{1}{2}}}{s+k^2}k^2 j_a \right)\frac{1}{H_0+i} \text{ } ds
\end{equation}

converges in norm to a bounded operator $L^2(\mathbb{R}^6)\rightarrow L^2(\mathbb{R}^6)$ (\cite{berger}). Then, by Lemma \ref{elemcomp}, each of the four operators in (\ref{integrand1})-(\ref{integrand4}) is in fact a compact operator for almost every $s$, so the integral (\ref{schwartz6}) converges to a compact operator. Consequently, the bounded operator $\frac{1}{H_0+i} \left( j_a |k| - |k| j_a \right)\frac{1}{H_0+i}$ from (\ref{schwartz0}) must extend uniquely from the dense Schwartz class to be this compact operator. This will conclude the proof for (\ref{compact6}).\bigskip

So, we prove that each of the operator norms $(\ref{integrand1})-(\ref{integrand4})$ are integrable functions of $s$. We can establish a lemma about Fourier multipliers that makes some of these computations simpler.

\begin{lemma}[An elementary boundedness lemma]\label{elembound}
Suppose that $f(q,r):\mathbb{R}^6\rightarrow \mathbb{C}$ and $g(x,y):\mathbb{R}^6\rightarrow \mathbb{C}$ are functions. Then, the operator norm of $f(q,r)g(x,y)$ as an operator on $L^2(\mathbb{R}^6)$ is bounded by $\Vert f \Vert_\ell \Vert g \Vert_\ell$ for any $2\le \ell \le \infty$, whenever these norms are finite.
\end{lemma}

The proof is omitted. Consider the operator in (\ref{integrand1}). Applying Lemma \ref{elembound},

\begin{align*}
\Vert \frac{1}{H_0+i} \frac{s^{-\frac{1}{2}}}{s+k^2} (\bigtriangleup_y j_a) \frac{1}{H_0+i}\Vert_{L^2\rightarrow L^2} &\le \Vert \frac{s^{-\frac{1}{2}}}{s+k^2} \Vert_\infty \Vert \bigtriangleup_y j_a \Vert_\infty \\
&\lesssim s^{-3/2}\\
\end{align*}

so at least this is integrable near $s=\infty$. Furthermore, we have for any $0<\epsilon<1$ and $0<\delta<1$ that

\begin{align*}
\frac{1}{H_0+i} \frac{s^{-\frac{1}{2}}}{s+k^2} (\bigtriangleup_y j_a) \frac{1}{H_0+i} &=  \frac{1}{H_0+i} \left( \frac{s^{-\frac{1}{2}}}{s+k^2} |k|^{1+\epsilon}\right) \left( \frac{1}{|k|^{1+\epsilon}} \frac{1}{(1+|y|^{1+\delta})}\right) \left( (1+|y|^{1+\delta})(\bigtriangleup_y j_a)\right)\frac{1}{H_0+i}\\
\end{align*}

as bounded operators. We fix (non-optimally) $\epsilon=\frac{1}{5}$ and $\delta=\frac{2}{5}$. Evidently $\left( (1+|y|^{1+\delta})(\bigtriangleup_y j_a)\right)$ is a bounded operator, because $\bigtriangleup_y j_a$ is homogeneous of degree $-2$. By H\"older's inequality, $\frac{1}{(1+|y|^{1+\delta})}$ is a bounded operator taking $L^2(\mathbb{R}^3_y)$ into $L^{10/9}(\mathbb{R}^3_y)$. Then by a Hardy-Littlewood-Sobolev estimate (Corollary 5.10 in \cite{liebloss}), $\frac{1}{|k|^{6/5}}$ is a bounded operator taking $L^{10/9}(\mathbb{R}^3_y)$ into $L^2(\mathbb{R}^3_y)$. Since  $\left( \frac{1}{|k|^{1+\epsilon}} \frac{1}{(1+|y|^{1+\delta})}\right)$ is therefore a bounded operator $L^2(\mathbb{R}^3_y)\mapsto L^2(\mathbb{R}^3_y)$, it extends to a bounded operator on the tensor product $L^2(\mathbb{R}^3_y \oplus \mathbb{R}^3_z)$. Lastly, the operator norm

$$\Vert \frac{s^{-\frac{1}{2}}}{s+k^2} |k|^{6/5} \Vert_{L^2(\mathbb{R}^6)\rightarrow L^2(\mathbb{R}^6)}  = \Vert \frac{s^{-\frac{1}{2}}}{s+k^2} |k|^{6/5} \Vert_\infty \lesssim s^{-9/10}$$

so that

\begin{align*}
\Vert \frac{1}{H_0+i} \frac{s^{-\frac{1}{2}}}{s+k^2} (\bigtriangleup_y j_a) \frac{1}{H_0+i} \Vert_{L^2(\mathbb{R}^6)\rightarrow L^2(\mathbb{R}^6)} & \lesssim s^{-9/10} \\
\end{align*}

Since the operator norm (\ref{integrand1}) is both $\lesssim s^{-3/2}$ and $\lesssim s^{-9/10}$, it is integrable in $s$.\bigskip

Now we turn to (\ref{integrand2}). We have

\begin{align*}
\Vert \frac{1}{H_0+i}\frac{s^{-\frac{1}{2}}}{s+k^2} (-ik \cdot (\bigtriangledown_y j_a))\Vert_{L^2(\mathbb{R}^6) \rightarrow L^2(\mathbb{R}^6)} &\le \sum_{\ell=1}^3 \Vert \frac{-ik_\ell}{i+p^2+|k|}\frac{s^{-\frac{1}{2}}}{s+k^2} (\bigtriangledown_y j_a)_\ell \Vert_{L^2(\mathbb{R}^6) \rightarrow L^2(\mathbb{R}^6)} \\
&\le \sum_{\ell=1}^3 \Vert \frac{-ik_\ell}{i+p^2+|k|} \Vert_{L^2(\mathbb{R}^6) \rightarrow L^2(\mathbb{R}^6)} \Vert \frac{s^{-\frac{1}{2}}}{s+k^2} (\frac{\partial}{\partial y_\ell} j_a) \Vert_{L^2(\mathbb{R}^6) \rightarrow L^2(\mathbb{R}^6)} \\
&\le \sum_{\ell=1}^3 \Vert \frac{-ik_\ell}{i+p^2+|k|} \Vert_{L^2(\mathbb{R}^6) \rightarrow L^2(\mathbb{R}^6)} \Vert \frac{s^{-\frac{1}{2}}}{s+k^2}\Vert_{L^\infty(\mathbb{R}^6)} \Vert (\frac{\partial}{\partial y_\ell} j_a) \Vert_{L^\infty(\mathbb{R}^6)} \\
&\lesssim s^{-3/2}
\end{align*}

where the second to last inequality uses Lemma \ref{elembound}. Now we need to deal with $s$ small. We obtain

\begin{align*}
\Vert \frac{1}{i+p^2+|k|}\frac{s^{-\frac{1}{2}}}{s+k^2} (-ik \cdot (\bigtriangledown_y j_a))\Vert_{L^2(\mathbb{R}^6) \rightarrow L^2(\mathbb{R}^6)} &\le \sum_{\ell=1}^3 \Vert \frac{1}{i+p^2+|k|}\frac{s^{-\frac{1}{2}}|k|}{s+k^2} (-i\frac{k_\ell}{|k|})(\frac{\partial}{\partial y_\ell} j_a) \Vert_{L^2(\mathbb{R}^6) \rightarrow L^2(\mathbb{R}^6)} \\
\end{align*}

Now for each fixed $x$, $\Vert (\frac{\partial}{\partial y_\ell} j_a) \psi(x,y) \Vert_{L^{14/9}(y)} \le  \Vert (\frac{\partial}{\partial y_\ell} j_a) \Vert_{L^{7}(y)} \Vert \psi(x,y) \Vert_{L^2(y)}\lesssim \Vert \psi(x,y) \Vert_{L^2(y)}$ with a constant that does not depend on $x$. The Fourier transform is bounded $L^{14/9} \rightarrow L^{14/5}$. Since $\frac{s^{-\frac{1}{2}}|k|}{s+k^2}$ is an operator $L^{14/5}(y)\rightarrow L^2(y)$ with operator norm $\Vert \frac{s^{-\frac{1}{2}}|k|}{s+k^2} \Vert_{L^7(k)} \lesssim s^{-11/14}$, the above operator norm $L^2(\mathbb{R}^6) \rightarrow L^2(\mathbb{R}^6)$ is $\lesssim s^{-11/14}$.\bigskip

 Thus (\ref{integrand2}) is also an integrable function of $s$. The proof that (\ref{integrand3}) and (\ref{integrand4}) are integrable functions of $s$ evidently reduces to the proofs for (\ref{integrand1}) and (\ref{integrand2}). Therefore the integral in (\ref{schwartz6}) converges to a compact operator, so (\ref{compact6}) is compact.\bigskip

Now we prove (\ref{compact5}) is compact. By the second resolvent identity
$$\left(\frac{1}{H+i}-\frac{1}{H_a+i}\right)j_a = \frac{1}{H+i}\left( I_a \right)\frac{1}{H_a+i}j_a$$
This is valid since $D(H_a)\subset D(I_a)$ by assumption. By commuting we get $$=\frac{1}{H+i}I_a j_a \frac{1}{H_a+i} + \frac{1}{H+i}I_a [\frac{1}{H_a+i},j_a]$$
These terms are both bounded operators a priori so this is a straightforward equality of bounded operators (We have $\frac{1}{H+i}\left( I_a \right)$ extending to a bounded operator by the relative boundedness assumptions on the potential). The first term $I_a j_a \frac{1}{H_a+i}$ is compact because of the relative compactness properties of our $I_aj_a$. The compactness of the second term $\frac{1}{H+i}I_a [\frac{1}{H_a+i},j_a]$ then reduces to compactness result (\ref{compact6}).\bigskip

Next, consider (\ref{compact7}). By relative boundedness assumptions it suffices to prove $\frac{1}{H_0+i}[[H,iA],j_a]\frac{1}{H_0+i}$ is compact. This is equal to $\frac{1}{H_0+i}[2p^2 + |k|,j_a]\frac{1}{H_0+i}$, which is compact by the proof for (\ref{compact5}).\bigskip

Finally, consider (\ref{compact8}). By relative boundedness assumptions, it suffices to prove $j_a [I_a,iA] \frac{1}{H_0+i}$ is compact, which is true from the relative compactness properties of the $j_a$. This concludes the proof of Lemma \ref{auxcompact}, so we turn to the proof of Lemma \ref{somecompact}.\bigskip

If a subset $\mathcal{F}\subseteq C_\infty(\mathbb{R})$, the continuous functions vanishing at infinity, has the following properties:

\begin{enumerate}
\item $\mathcal{F}$ contains resolvents $\frac{1}{x+\xi}$ for all $\xi$ in some open set $u \subset \mathbb{C}$.
\item $\mathcal{F}$ forms a vector space.
\item $\mathcal{F}$ is closed under convergence in the $L^\infty$ norm.
\end{enumerate}

Then $\mathcal{F} = C_\infty(\mathbb{R})$ (Appendix to ch. 3 in \cite{cfks}).\bigskip

Let $\mathcal{F}$ be the class of such functions $f$ so that (\ref{compact2}) is compact. We have proven the first property in Lemma \ref{auxcompact} (the choice of $\xi=i$ was arbitrary), and the second and third properties are evident. Since $f$ was taken to be $C^\infty_0$ in Lemma \ref{ims}, this is sufficient for (\ref{compact2}).\bigskip

Now let $\mathcal{F}$ be the class of such functions $f$ so that (\ref{compact1}) is compact. Again, we have proven the first property, and the second and third properties are evident.\bigskip

For (\ref{compact3}) and (\ref{compact4}), one can prove compactness by multiplying and dividing by resolvents to reduce to (\ref{compact7}) and (\ref{compact8}).\bigskip

Thus, we do indeed have the breaking apart of the main estimate as in (\ref{ims}).

\subsection{The cluster $(xy0)$}

Directly from the previous section we have that $f(H)j_{(xy0)}$ is compact:

\begin{align*}
f(H)j_{(xy0)} &= f(H)(H+i)\frac{1}{H+i}(H_0+i) \frac{1}{H_0+i} j_{(xy0)} \\
\end{align*}

which is compact since $f(H)(H+i)$ is bounded, $\frac{1}{H+i}(H_0+i)$ is bounded, and $ \frac{1}{H_0+i} j_{(xy0)}$ is compact by Lemma \ref{elemcomp}. Therefore $j_{(xy0)} f(H) [H,iA] f(H) j_{(xy0)} $ is compact as claimed in (\ref{ims}).

\subsection{The cluster $(x)(y)(0)$}

In this section, we fix $a=(x)(y)(0)$, the cluster decomposition corresponding to free dynamics.

\begin{lemma}\label{alpha1}
Fix $\epsilon>0$ and an energy $E \neq 0$. Then there exists $\delta>0$ so that $H_a$ satisfies a Mourre estimate at $E$ with conjugate operator $A$, width $\delta$, and constant $\alpha_{(x)(y)(0)}$, where $$\alpha_{(x)(y)(0)}=\begin{cases}
E-\epsilon: \text{ } E>0 \\
c: \text{ } E<0 \\
\end{cases}$$
and $c$ is any positive constant one wishes.
\end{lemma}

\begin{proof}[Proof for $E<0$]
Since $E_\bigtriangleup(H_a)=0$ the desired operator inequality is trivial.
\end{proof}

\begin{proof}[Proof for $E>0$]

Fix $\epsilon>0$. Pick $\delta< \min(E,\epsilon)$, and select $\bigtriangleup = (E-\delta,E+\delta)$. We have
\begin{align*}
E_\bigtriangleup (H_a)[H_a,iA]E_\bigtriangleup(H_a)&= E_\bigtriangleup(p^2+|k|)\left(2p^2 + |k|\right)E_\bigtriangleup(p^2+|k|)\\
\ge (E-\delta) E_\bigtriangleup(p^2+|k|)\\
\ge (E-\epsilon) E_\bigtriangleup(p^2+|k|)\\
\end{align*}
from the functional calculus.
\end{proof}

\subsection{The cluster $(x)(y0)$}

In this section, fix $a=(x)(y0)$, the cluster decomposition corresponding to the photon-proton cluster. Our aim is to prove the following.

\begin{lemma}\label{alpha2}

Fix $\epsilon>0$ and an energy $E\neq 0$ not an eigenvalue of the subsystem Hamiltonian $h_a$. Then there exists $\delta>0$ so that $H_a$ satisfies a Mourre estimate at $E$ with conjugate operator $A$, width $\delta$, and constant $\alpha_{(x)(y0)}$, where
$$\alpha_{(x)(y0)}=
\begin{cases}
\min(2d(E,a)-\epsilon,E-\epsilon): \text{ } E>0 \\
2d(E,a)-\epsilon: \text{ } E<0 \\
\end{cases}$$
\end{lemma}

The strategy of proof is to consider each value of the electron's relative momentum $r$ separately. Because $r$ commutes with the operator $[H_a,iA]$, we can defined the \textbf{fibered commutator}

$$[H_a,iA](s) := 2s^2 + |k| - y\cdot \bigtriangledown V_{13}(y)$$

which is a self-adjoint operator on $L^2(\mathbb{R}^3_y)$ by Kato-Rellich. We can then write as a direct integral:

\begin{align*}
E_\bigtriangleup (H_a) [H_a,iA] E_\bigtriangleup (H_a) &= E_\bigtriangleup (p^2 + |k| + V_{13}(y)) \left( 2p^2 +|k| - y\cdot \bigtriangledown_y V_{13}(y)\right)  E_\bigtriangleup (p^2 + |k| + V_{13}(y))\\
&= \int^\oplus_{\mathbb{R}^3_s} E_\bigtriangleup (s^2 + |k| + V_{13}(y)) \left( 2s^2 +|k| - y\cdot \bigtriangledown_y V_{13}(y)\right)  E_\bigtriangleup (s^2 + |k| + V_{13}(y)) ds \\
&= \int^\oplus_{\mathbb{R}^3_s} E_\bigtriangleup (H_a(s)) [H_a,iA](s) E_\bigtriangleup (H_a(s)) ds \\
\end{align*}

It is then sufficient to prove a uniform Mourre estimate on each fiber. For fibers where the relative momentum $s$ is large, the positive contribution to the commutator comes from $2s^2$. For fibers where the relative momentum is small, we rely on the energy $E$ being away from thresholds for the positive contribution. With this in mind, we prove  Lemma \ref{alpha2lemma1} and Lemma \ref{alpha2lemma2}, which show that the remaining terms are small.

\begin{lemma}\label{alpha2lemma1}
Fix $\epsilon>0$, an energy $E\neq0$ not an eigenvalue of the subsystem Hamiltonian $h_a$, and a value of $s_0\neq 0$. Then there exists $\delta>0$ and an open set $U\subset \mathbb{R}^3$ containing $s_0$ so that for all $s \in U$, letting $\bigtriangleup = (E-\delta,E+\delta)$:

$$E_\bigtriangleup (H_a(s)) [h_a, iA^a] E_\bigtriangleup (H_a(s)) \ge -\epsilon E_\bigtriangleup (H_a(s))$$

\end{lemma}

\begin{proof}

Fix $\epsilon$, $E$ and $s_0$ as above. We may select $\delta_0>0$ so that for $\bigtriangleup_0=(E-\delta_0,E+\delta_0)$:

\begin{equation}\label{BStrick1}
E_{\bigtriangleup_0} (H_a(s_0)) [h_a, iA^a] E_{\bigtriangleup_0} (H_a(s_0)) \ge -\frac{\epsilon}{4}E_{\bigtriangleup_0} (H_a(s_0)) + K
\end{equation}

where $K$ is a compact operator. To see why, we break up into three cases. First, if $E-s_0^2$ is a negative eigenvalue of $h_a$, then choose $\delta_0$ so small that $\bigtriangleup_0-s_0^2$ contains no continuous spectrum of $h_a$. Then by the virial theorem for $[h_a,iA^a]$ we have

\begin{align*}
E_\bigtriangleup (H_a(s_0)) [h_a, iA^a] E_\bigtriangleup (H_a(s_0)) &= E_\bigtriangleup (s_0^2 + |k| + V_{13}(y)) [h_a, iA^a] E_\bigtriangleup (s_0^2 + |k| + V_{13}(y))\\
&= E_{\bigtriangleup-s_0^2} (h_a)[h_a, iA^a] E_{\bigtriangleup-s_0^2} (h_a)\\
&=0\\
\end{align*}

Second, if $E-s_0^2$ is negative but in the resolvent set of $h_a$, we can pick $\delta_0$ so small that $\bigtriangleup_0-s_0^2$ contains no spectrum of $h_a$ at all, and the desired estimate (\ref{BStrick1}) follows because the projections $E_{\bigtriangleup-s_0^2} (h_a)$ are zero. Third, if $E-s_0^2$ is nonnegative, we can take $0 < \delta_0<  E-s_0^2+\frac{\epsilon}{4}$. Then we compute

\begin{align*}
E_{\bigtriangleup_0} (H_a(s_0)) [h_a, iA^a] E_{\bigtriangleup_0} (H_a(s_0)) &= E_{\bigtriangleup_0} (s_0^2 + |k| + V_{13}(y)) \Bigg( |k|- y\cdot \bigtriangledown V_{13}(y) \Bigg) E_{\bigtriangleup_0} (s_0^2 + |k| + V_{13}(y))\\
&= E_{\bigtriangleup_0-s_0^2} (|k| + V_{13}(y)) \Bigg( |k|- y\cdot \bigtriangledown V_{13}(y) \Bigg) E_{\bigtriangleup_0-s_0^2} (|k| + V_{13}(y))\\
&= E_{\bigtriangleup_0-s_0^2} (|k| + V_{13}(y)) \Bigg( |k|+V_{13}(y)\Bigg)E_{\bigtriangleup_0-s_0^2} (|k| + V_{13}(y))\\
&\hspace{.5cm}+ E_{\bigtriangleup_0-s_0^2} (|k| + V_{13}(y))\Bigg( -V_{13}(y)- y\cdot \bigtriangledown V_{13}(y) \Bigg) E_{\bigtriangleup_0-s_0^2} (|k| + V_{13}(y))\\
&= (E-s_0^2-\frac{\epsilon}{4})E_{\bigtriangleup_0-s_0^2} (|k| + V_{13}(y))  + K\\
&\ge -\frac{\epsilon}{4}E_{\bigtriangleup_0} (H_a(s_0)) + K\\
\end{align*}

by using the functional calculus, where $K= E_{\bigtriangleup_0-s_0^2} (h_a)\Bigg( -V_{13}(y)- y\cdot \bigtriangledown V_{13}(y) \Bigg) E_{\bigtriangleup-s_0^2} (h_a)$ is compact. Thus we can always find $\delta_0$ so that (\ref{BStrick1}) holds. All that remains is to prove Lemma \ref{alpha2lemma1} is to remove the compact $K$. For another proof in this vein, see (Eqn. (3.4) in \cite{froeseherbst}).\bigskip

Fix such a $\delta_0$ as in (\ref{BStrick1}). We write e.g. $E_{pp {\bigtriangleup_0}}(H_a(s_0)):= E_{pp}(H_a(s_0))  E_{\bigtriangleup_0} (H_a(s_0)) $, where $E_{pp}(H_a(s_0))$ represents the projection onto the pure point spectrum. The next claim is that we have:

\begin{equation}\label{BStrick2}
E_{\bigtriangleup_0}(H_a(s_0)) [h_a,iA^a] E_{\bigtriangleup_0}(H_a(s_0))  \ge -\frac{\epsilon}{2} E_{\bigtriangleup_0}(H_a(s_0)) + (1- E_{pp\bigtriangleup_0}(H_a(s_0))) K_1  (1- E_{pp\bigtriangleup_0}(H_a(s_0)))
\end{equation}

for some compact $K_1$. Assuming (\ref{BStrick2}), then we could select $\delta_1$ so small that for $\bigtriangleup_1 =(E-\delta_1, E+\delta_1)$, we could multiply both sides of the above by $E_{\bigtriangleup_1}(H_a(s_0))$ and use

$$(E_{\bigtriangleup_1}(H_a(s_0))- E_{pp\bigtriangleup_1}(H_a(s_0))) K_1  (E_{\bigtriangleup_1}(H_a(s_0))- E_{pp\bigtriangleup_1}(H_a(s_0)))  \ge -\frac{\epsilon}{2} E_{\bigtriangleup_1}(H_a(s_0))$$

to conclude that:

\begin{equation}\label{BStrick3}
E_{\bigtriangleup_1}(H_a(s_0)) [h_a,iA^a] E_{\bigtriangleup_1}(H_a(s_0))  \ge -\epsilon E_{\bigtriangleup_1}(H_a(s_0))
\end{equation}

Then we would select $\delta= \delta_1/2$ and define $U:= \lbrace s\in \mathbb{R}^3 : s_0^2-\delta < s^2 < s_0^2+\delta\rbrace$. This would give us the conclusion of the lemma, letting $\bigtriangleup= (E-\delta,E+\delta)$; for any $s\in U$, we could prove the desired inequality by taking (\ref{BStrick3}) and multiplying on both sides by $E_{\bigtriangleup-s^2}(h_a)$. This is due to the fact that as defined, $\bigtriangleup-s^2 \subset \bigtriangleup_1-s_0^2$ for any $s\in U$. So it remains to show (\ref{BStrick2}).\bigskip

Since $K$ is compact, we may select a finite-dimensional projection $F$ with range contained in that of $E_{\bigtriangleup_0 pp}(H_a(s_0))$ so that

$$\Vert (1-F) K (1-F)- (1-E_{\bigtriangleup_0 pp}(H_a(s_0))) K (1-E_{\bigtriangleup_0 pp}(H_a(s_0)))\Vert\le \frac{\epsilon}{2}$$

Then, from multiplying (\ref{BStrick1}) on both sides by $(1-F)$, we obtain

$$(E_{\bigtriangleup_0}(H_a(s_0))-F) [h_a,iA^a] (E_{\bigtriangleup_0}(H_a(s_0))-F) \ge -\frac{\epsilon}{4}(E_{\bigtriangleup_0 pp}(H_a(s_0))-F) +(1-F)K(1-F)$$

By using the property by which $F$ was obtained, we glean from this that

\begin{align*}
&(E_{\bigtriangleup_0}(H_a(s_0))-F) [h_a,iA^a] (E_{\bigtriangleup_0}(H_a(s_0))-F) \\
&\ge -\frac{\epsilon}{2}(E_{\bigtriangleup_0 pp}(H_a(s_0))-F) +(1-E_{\bigtriangleup_0 pp}(H_a(s_0)))K(1-E_{\bigtriangleup_0 pp}(H_a(s_0)))\\
\end{align*}

We can multiply out the left-hand side of the above, and apply the virial theorem to get

\begin{align*}
&E_{\bigtriangleup_0}(H_a(s_0)) [h_a,iA^a] E_{\bigtriangleup_0}(H_a(s_0)) -  C^*F- F^*C \\
&\ge -\frac{\epsilon}{2}(E_{\bigtriangleup_0 pp}(H_a(s_0))-F) +(1-E_{\bigtriangleup_0 pp}(H_a(s_0))K(1-E_{\bigtriangleup_0 pp}(H_a(s_0)) \\
\end{align*}

where we have let $C = F [h_a,iA^a]E_{\bigtriangleup_0}(H_a(s_0))(1-E_{\bigtriangleup_0 pp}(H_a(s_0)))$. To obtain (\ref{BStrick2}) from this is a matter of showing that for some compact $K_2$,

$$ C^*F+ F^*C \ge (1-E_{\bigtriangleup_0 pp}) K_2 (1-E_{\bigtriangleup_0 pp}) + \frac{\epsilon}{2}F^* F$$

Yet this follows from

$$C^*F+ F^*C \ge - (\frac{2}{\epsilon}C^* C + \frac{\epsilon}{2}F^* F)$$

letting $K_2= -\frac{2}{\epsilon}C^*C$. So let $K_1= (1-E_{\bigtriangleup_0 pp})(K+K_2)(1-E_{\bigtriangleup_0 pp})$. Note that it doesn't matter how $K_1$ depends on $\epsilon$, because of the discussion surrounding (\ref{BStrick2}) and (\ref{BStrick3}). That completes the analysis.
\end{proof}

We will also use the following.

\begin{lemma}\label{alpha2lemma2}
Fix $\epsilon>0$, an energy $E\neq 0$ not an eigenvalue of the subsystem Hamiltonian $h_a$, and a value of $s_0\neq 0$ so that $E-s_0^2$ is not an eigenvalue of $h_a$. Then there exists $\delta>0$ and an open set $U\subset \mathbb{R}^3$ containing $s_0$ so that for all $s \in U$, letting $\bigtriangleup = (E-\delta,E+\delta)$:

$$E_\bigtriangleup (H_a(s)) (-V_{13}(y)-y\cdot \bigtriangledown V_{13}(y))  E_\bigtriangleup (H_a(s))  \ge -\epsilon E_\bigtriangleup (H_a(s)) $$

\end{lemma}

\begin{proof}
Fix $\epsilon$, $E$, and $s_0$ as above. Then there is a width $\delta_0>0$ so that $(E-\delta_0 - s_0^2, E+\delta_0 - s_0^2)$ contains no eigenvalues of $h_a$. Then, there exists $\delta_1\le \delta_0$ small enough so that for $\bigtriangleup_1 = (E-\delta_1,E+\delta_1)$:

$$E_{\bigtriangleup_1} (s_0^2 + |k|+V_{13}(y) (-V_{13}(y)-y\cdot \bigtriangledown V_{13}(y))  E_{\bigtriangleup_1} (s_0^2 + |k|+V_{13}(y)) \ge -\epsilon E_{\bigtriangleup_1} (s_0^2 + |k|+V_{13}(y))$$

We select $\delta= \delta_1/2$, and let $U:= \lbrace s\in \mathbb{R}^3 : s_0^2-\delta < s^2 < s_0^2+\delta\rbrace$. As before, this concludes the proof.
\end{proof}

These two lemmas give us sufficient control over the junk terms in the fibered commutator $[H_a,iA](s)$, so we can now attack it.\bigskip

\begin{lemma}\label{alpha2Uform}

Fix $\epsilon>0$, an energy $E\neq 0$ not an eigenvalue of the subsystem Hamiltonian $h_a$ and a value $s_0\in \mathbb{R}^3$. Then there exists $\delta>0$ and an open set $U$ containing $s_0$ so that for all $s \in U$, letting $\bigtriangleup = (E-\delta,E+\delta)$ we have:

\begin{align*}
E_\bigtriangleup (H_a(s)) [H_a,iA](s) E_\bigtriangleup (H_a(s)) &\ge \alpha_{(x)(y0)} E_\bigtriangleup (H_a(s))\\
\end{align*}

Note that the dependence on $\epsilon$ comes from the way $\alpha_{(x)(y0)}$ was defined.
\end{lemma}

For the proof, we fix $\epsilon$, $E$ and $s_0$ as above. Select $\delta_0$ so that $\bigtriangleup_0 = [E-\delta_0, E+\delta_0]$ does not contain $0$ or any thresholds, and also so $\delta_0 < \epsilon/2$. Then select $\tau>0$ so that $[E-\delta_0-t, E+\delta_0-t]$ does not contain $0$ or any thresholds for $0\le t\le \tau$. Furthermore $\tau$ can be selected so $\tau \ge d(E,a)-\epsilon$. The choice of $\tau$ serves to separate our analysis into `small external momentum' and `large external momentum'.\bigskip

We handle the cases of $E<0$ and $E>0$ separately.\bigskip

\begin{proof}[Proof for $E<0$]
Consider the small-momentum case, $s_0^2\le \tau$. The projection $E_{\bigtriangleup_0}(H_a(s_0))$ is evidently zero. We may select $\delta=\delta_0/2$ and $U:= \lbrace s\in \mathbb{R}^3 : s_0^2-\delta < s^2 < s_0^2+\delta\rbrace$. Then for all $s\in U$, letting $\bigtriangleup = (E-\delta,E+\delta)$:

\begin{align*}
E_\bigtriangleup (H_a(s)) [H_a,iA](s) E_\bigtriangleup (H_a(s))\ge (2\tau-\epsilon)  E_\bigtriangleup (H_a(s))\\
\end{align*}

for all $s \in U$, since the projections $E_\bigtriangleup(H_a(s))$ are zero for all such $s$, so we may have in fact any constant we wish (in place of $2\tau-\epsilon$).\bigskip

Now consider the large-momentum case, $s_0^2\ge \tau$. By Lemma \ref{alpha2lemma1}, we may select $\delta<\delta_0$ and $U$ containing $s_0$ so that for all $s \in U$, letting $\bigtriangleup= (E-\delta, E+\delta)$:

$$E_\bigtriangleup (H_a(s)) [h_a,iA^a] E_\bigtriangleup (H_a(s)) \ge -\epsilon E_\bigtriangleup (H_a(s))$$

Then, we have

\begin{align*}
E_\bigtriangleup (H_a(s)) [H_a,iA](s) E_\bigtriangleup (H_a(s)) &= E_\bigtriangleup (H_a(s)) \Bigg( 2s^2 + [h_a,iA^a] \Bigg) E_\bigtriangleup (H_a(s))\\
&\ge 2 \tau E_\bigtriangleup (H_a(s)) - \epsilon E_\bigtriangleup (H_a(s))\\
&\ge (2\tau-\epsilon)E_\bigtriangleup (H_a(s))\\
\end{align*}

for all $s\in U$.\bigskip

So, for any fiber $s_0$: there exists $\delta>0$ and a $U$ containing $s_0$ so that for all $s \in U$, letting $\bigtriangleup= (E-\delta, E+\delta)$:

\begin{align*}
E_\bigtriangleup (H_a(s)) [H_a,iA](s) E_\bigtriangleup (H_a(s)) &\ge (2\tau-\epsilon)E_\bigtriangleup (H_a(s))\\
&\ge (2d(E,a) - 3\epsilon)E_\bigtriangleup (H_a(s))\\
\end{align*}

where the last line holds because $\tau \ge d(E,a)-\epsilon$. By a renaming of $\epsilon$ we have our conclusion.
\end{proof}

\begin{proof}[Proof for $E>0$]

Consider the small-momentum case, $s_0^2\le \tau$.  By Lemma \ref{alpha2lemma2}, we may select $\delta<\delta_0$ and $U$ containing $s_0$ so that for all $s\in U$, letting $\bigtriangleup= (E-\delta, E+\delta)$:

$$E_\bigtriangleup (H_a(s)) (-V_{13}(y)-y\cdot \bigtriangledown V_{13}(y))  E_\bigtriangleup (H_a(s)) \ge -\frac{\epsilon}{2} E_\bigtriangleup (H_a(s))$$

Then, we have

\begin{align*}
E_\bigtriangleup (H_a(s)) [H_a,iA](s) E_\bigtriangleup (H_a(s)) &= E_\bigtriangleup (H_a(s)) \Bigg( 2s^2 + |k| + V_{13}(y) - V_{13}(y) - y\cdot \bigtriangledown V_{13}(y) \Bigg) E_\bigtriangleup (H_a(s))\\
&\ge (E-\delta_0) E_\bigtriangleup (H_a(s)) + E_\bigtriangleup (H_a(s)) \Bigg( - V_{13}(y) - y\cdot \bigtriangledown V_{13}(y) \Bigg) E_\bigtriangleup (H_a(s))\\
&\ge (E-\delta_0 - \frac{\epsilon}{2}) E_\bigtriangleup (H_a(s))\\
&\ge (E-\epsilon)  E_\bigtriangleup (H_a(s))\\
\end{align*}

for all $s \in U$, where the second-to-last step is by the functional calculus..\bigskip

The large-momentum case for $E>0$ is handled the same way as the large-momentum case for $E<0$; the same estimate with constant $(2\tau-\epsilon)$ holds.\bigskip

So, for any fiber $s_0^2$: there exists a $\delta>0$ and a $U$ containing $s_0$ so that for all $s \in U$, letting $\bigtriangleup=(E-\delta,E+\delta)$:

\begin{align*}
E_\bigtriangleup (H_a(s))[H_a,iA](s) E_\bigtriangleup (H_a(s))&\ge \min(E-\epsilon,2\tau-\epsilon) E_\bigtriangleup (H_a(s))\\
&\ge \min(E-\epsilon,2d(E,a) - 3\epsilon)E_\bigtriangleup (H_a(s))\\
\end{align*}

where the last line holds because $\tau \ge d(E,a)-\epsilon$. By a renaming of $\epsilon$ we have our conclusion.
\end{proof}

Now we can proceed with the proof of Lemma \ref{alpha2}.

\begin{proof}
Fix $\epsilon>0$ and an energy $E$. If $E<G_a$, then by taking $\delta$ small enough, the projections $E_\bigtriangleup(H_a)$ are zero and the conclusion is trivial, so we may assume $E\ge G_a$. Let $M$ be a number so that $M>> E-G_a$. Take the compact set $\lbrace s\in \mathbb{R}^3 : s^2\le M \rbrace$ and use Lemma \ref{alpha2Uform} to cover it with sets $U_i$, so that there exists a $\delta_i>0$ so that for each $s \in U_i$, taking $\bigtriangleup_i=(E-\delta_i,E+\delta_i)$:

$$E_{\bigtriangleup_i}(H_a(s)) [H_a,iA](s) E_{\bigtriangleup_i}(H_a(s)) \ge \alpha_{(x)(y0)}E_{\bigtriangleup_i}(H_a(s))$$

Extract a finite subcover and let $\delta$ be the minimum over the finite collection of $\delta_i$ associated to the subcover. It is then the case that for all $s$ such that $s^2 \le M$, taking $\bigtriangleup = (E-\delta,E+\delta)$:

$$E_{\bigtriangleup}(H_a(s)) [H_a,iA](s) E_{\bigtriangleup}(H_a(s)) \ge \alpha_{(x)(y0)}E_{\bigtriangleup}(H_a(s))$$

Since the projections are $0$ when $s^2 >  E-G_a$, the above inequality is also true for $s^2> M$. We conclude that the inequality holds for all $s \in \mathbb{R}^3$. Thus we have an inequality on the whole direct integral

$$\int^\oplus_{\mathbb{R}^3_s}  E_\bigtriangleup (H_a(s)) [H_a,iA](s) E_\bigtriangleup (H_a(s)) ds \ge \alpha_{(x)(y0)} \int^\oplus_{\mathbb{R}^3_s} E_\bigtriangleup (H_a(s)) ds$$

which is exactly Lemma \ref{alpha2}.\bigskip
\end{proof}

\subsection{The cluster $(y)(x0)$}

In what follows, fix $a=(y)(x0)$, the cluster decomposition corresponding to the photon-proton cluster. The analysis is much the same as the previous cluster but is outlined for the sake of completeness. Our aim is to prove the following.

\begin{lemma}\label{alpha3}

Fix $\epsilon>0$ and an energy $E\neq 0$ not an eigenvalue of the subsystem Hamiltonian $h_a$. Then there exists $\delta>0$ so that $H_a$ satisfies a Mourre estimate at $E$ with conjugate operator $A$, width $\delta$, and constant $\alpha_{(y)(x0)}$, where

$$\alpha_{(y)(x0)}=
\begin{cases}
\min(d(E,a)-\epsilon,E-\epsilon): \text{ } E>0 \\
d(E,a)-\epsilon: \text{ } E<0 \\
\end{cases}$$
\end{lemma}

Because $q$ commutes with the operator $[H_a,iA]$, we can define the fibered commutator

$$[H_a,iA](s) = |s| + 2p^2 - x\cdot \bigtriangledown V_{12}(x)$$

and then writing as a direct integral:

\begin{align*}
E_\bigtriangleup (H_a) [H_a,iA] E_\bigtriangleup (H_a) &= E_\bigtriangleup (p^2 + |k| + V_{12}(x)) \left( 2p^2 +|k| - x\cdot \bigtriangledown V_{12}(x)\right)  E_\bigtriangleup (p^2 + |k| + V_{12}(x))\\
&= \int^\oplus_{\mathbb{R}^3_s} E_\bigtriangleup (|s| +p^2 + V_{12}(x)) \left( |s|+  2p^2 - x\cdot \bigtriangledown V_{12}(x)\right)  E_\bigtriangleup (|s|+ p^2 + V_{12}(x)) ds \\
&= \int^\oplus_{\mathbb{R}^3_s} E_\bigtriangleup (H_a(s)) [H_a,iA](s) E_\bigtriangleup (H_a(s)) ds\\
\end{align*}

It will be sufficient to prove a Mourre estimate of each fiber. The positive contribution will come from $|s|$ for large values of $s$, and from the energy $E$ being nonthreshold for small values of $s$. Lemmas \ref{alpha3lemma1} and \ref{alpha3lemma2} show that the remaining terms are small.

\begin{lemma}\label{alpha3lemma1}
Fix $\epsilon>0$, an energy $E\neq 0$ not an eigenvalue of the subsystem Hamiltonian $h_a$, and a value of $s_0\neq 0$, there exists $\delta=\delta(s_0,\epsilon)>0$ and and open set $U\subset \mathbb{R}^3$ containing $s_0$ so that for all $s \in U$, letting $\bigtriangleup = (E-\delta,E+\delta)$:

$$E_\bigtriangleup (H_a(s)) [h_a,iA^a] E_\bigtriangleup (H_a(s)) \ge -\epsilon E_\bigtriangleup (H_a(s))$$

\end{lemma}

\begin{proof}

Fix $\epsilon$, $E$, and $s_0$ as above. We may select $\delta_0>0$ so that for $\bigtriangleup_0 = (E-\delta_0, E+\delta_0)$:

\begin{equation}\label{2BStrick1}
E_{\bigtriangleup_0}(H_a(s_0)) [h_a,iA^a] E_{\bigtriangleup_0} (H_a(s_0)) \ge -\frac{\epsilon}{4} E_{\bigtriangleup_0} (H_a(s_0)) + K
\end{equation}

where $K$ is a compact operator. To see why, break up into three cases. If $E-|s_0|$ is a negative eigenvalue of $h_a$, then choose $\delta_0$ so small that $\bigtriangleup_0-|s_0|$ contains no continuous spectrum of $h_a$. Then by the virial theorem for $[h_a,iA^a]$, we have

$$E_{\bigtriangleup_0}(H_a(s_0)) [h_a,iA^a] E_{\bigtriangleup_0} (H_a(s_0))=0$$

Second, if $E-|s_0|$ is negative but in the resolvent set of $h_a$, we can pick $\delta_0$ so small that $\bigtriangleup_0-|s_0|$ contains no spectrum of $h_a$ at all, whereby the desired estimate (\ref{2BStrick1}) follows because the projections $E_{\bigtriangleup_0}(H_a(s_0))$ are zero. Third, if $\bigtriangleup_0-|s_0|$ is nonnegative, we can take $0\le \delta_0 < E-|s_0|+\frac{\epsilon}{4}$. Then we compute

\begin{align*}
E_{\bigtriangleup_0}(H_a(s_0))[h_a,iA^a] E_{\bigtriangleup_0}(H_a(s_0)) &= E_{\bigtriangleup_0 - |s_0|}(p^2+V_{12}(x))\Bigg( p^2 + V_{12}(x) \Bigg) E_{\bigtriangleup_0 - |s_0|}(p^2+V_{12}(x))\\
& \hspace{.5cm} + E_{\bigtriangleup_0 - |s_0|}(p^2+V_{12}(x))\Bigg( - V_{12}(x) - x \cdot \bigtriangledown V_{12}(x) \Bigg) E_{\bigtriangleup_0 - |s_0|}(p^2+V_{12}(x))\\
& \ge (E-|s_0|-\frac{\epsilon}{4}) E_{\bigtriangleup_0 - |s_0|}(p^2+V_{12}(x)) + K\\
&\ge -\frac{\epsilon}{4} E_{\bigtriangleup_0 - |s_0|}(p^2+V_{12}(x)) + K\\
\end{align*}

by using the functional calculus, where $K=E_{\bigtriangleup_0 - |s_0|}(h_a)\Bigg( - V_{12}(x) - x \cdot \bigtriangledown V_{12}(x) \Bigg) E_{\bigtriangleup_0 - |s_0|}(h_a)$ is compact. Thus we can always find $\delta_0$ so that (\ref{2BStrick1}) holds. It remains to remove the compact $K$.\bigskip

Fix such a $\delta_0$ as in (\ref{2BStrick1}). We claim that

\begin{equation}\label{2BStrick2}
E_{\bigtriangleup_0}(H_a(s_0)) [h_a, iA^a] E_{\bigtriangleup_0}(H_a(s_0)) \ge -\frac{\epsilon}{2} E_{\bigtriangleup_0}(H_a(s_0)) + (1-E_{pp \bigtriangleup_0}(H_a(s_0)))K_1 (1-E_{pp \bigtriangleup_0}(H_a(s_0)))
\end{equation}

for some compact $K_1$. Assuming (\ref{2BStrick2}), we could then select $\delta_1$ so small that for $\bigtriangleup_1 = (E-\delta_1,E+\delta_1)$, we could multiply both sies of the above by $E_{\bigtriangleup_1}(H_a(s_0))$ and use

$$(E_{\bigtriangleup_1}(H_a(s_0))-E_{pp \bigtriangleup_1}(H_a(s_0)))K_1 (E_{\bigtriangleup_1}(H_a(s_0))-E_{pp \bigtriangleup_1}(H_a(s_0)))\ge -\frac{\epsilon}{2}E_{\bigtriangleup_1}(H_a(s_0))$$

to conclude that

\begin{equation}\label{2BStrick3}
E_{\bigtriangleup_1}(H_a(s_0)) [h_a(y),iA^a] E_{\bigtriangleup_1} (H_a(s_0)) \ge -\epsilon E_{\bigtriangleup_1}(H_a(s_0))
\end{equation}

Then we would select $\delta= \delta_1/2$ and define $U=: \lbrace s\in \mathbb{R}^3: |s_0|-\delta < |s| < |s_0|+\delta$. This would give us the conclusion of the lemma, letting $\bigtriangleup=(E-\delta,E+\delta)$; for any $s\in U$, we could prove the desired inequality by taking (\ref{2BStrick3}) and multiplying on both sides by $E_{\bigtriangleup-|s|}(h_a)$. So it remains to show (\ref{2BStrick2}). But this follows from the same argument as the proof of (\ref{BStrick2}) from the previous section.
\end{proof}

We will also use the following.

\begin{lemma}\label{alpha3lemma2}
Fix $\epsilon>0$, an energy $E\neq 0$ not an eigenvalue of the subsystem Hamiltonian $h_a$, and a value of $s_0\neq 0$ so that $E-|s_0|$ is not an eigenvalue of $h_a$. Then there exists $\delta>0$ and an open set $U$ containing $s_0$ so that for all $s_0 \in U$, letting $\bigtriangleup = (E-\delta,E+\delta)$:

$$E_\bigtriangleup (H_a(s)) + (-V_{12}(x)-x\cdot \bigtriangledown V_{12}(x))  E_\bigtriangleup ((H_a(s)) \ge -\epsilon E_\bigtriangleup ((H_a(s))$$

\end{lemma}

\begin{proof}
Fix $\epsilon$, $E$, and $s_0$ as above. Then there is a width $\delta_0>0$ so that $(E-\delta_0 - |s_0|, E+\delta_0 -|s_0|)$ contains no eigenvalues of $h_a$. Then, there exists $\delta_1\le \delta_0$ so that for $\bigtriangleup_1 = (E-\delta_1,E+\delta_1)$:

$$E_{\bigtriangleup_1} (|s|+p^2+V_{12}(x)+ (-V_{12}(x)-x\cdot \bigtriangledown V_{12}(x))  E_{\bigtriangleup_1} (|s|+p^2+V_{12}(x)) \ge -\epsilon E_{\bigtriangleup_1} (p^2+|s|+V_{12}(x))$$

We select $\delta= \delta_1/2$, and let $U= \lbrace s\in \mathbb{R}^3: |s_0|-\delta < |s| < |s_0|+\delta\rbrace$.
\end{proof}

\bigskip

\begin{lemma}\label{alpha3Uform}

Fix $\epsilon>0$, an energy $E\neq 0$ not an eigenvalue of the subsystem Hamiltonian $h_a$, and a value of $s_0 \neq 0$. Then there exists a $\delta >0$ and an open set $U\subset \mathbb{R}^3$ containing $s_0$ so that for all $s \in U$, letting $\bigtriangleup = (E-\delta,E+\delta)$:

\begin{align*}
E_\bigtriangleup (H_a(s))[H_a,iA](s) E_\bigtriangleup (H_a(s)) &\ge \alpha_{(y)(x0)} E_\bigtriangleup (H_a(s))\\
\end{align*}

\end{lemma}

Fix $\epsilon>0$, $E$, and $s_0$ as above. Select $\delta_0$ so that $\bigtriangleup_0 = [E-\delta_0, E+\delta_0]$ does not contain $0$ or any thresholds, and also so $\delta_0\epsilon/2$. Then select $\tau>0$ so that $[E-\delta_0-t, E+\delta_0-t]$ does not contain $0$ or any thresholds for $0\le t\le \tau$. Furthermore $\tau$ can be selected to be $\ge d(E,a)-\epsilon$.\bigskip

We handle the cases of $E<0$ and $E>0$ separately.\bigskip

\begin{proof}[Proof for $E<0$]
Now consider $E<0$.\bigskip

Consider $|s_0|\le \tau$. The projection $E_{\bigtriangleup_0}(H_a(s_0))$ is evidently $0$. We may select $\delta=\delta_0$ and $U := \lbrace s\in \mathbb{R}^3: |s_0|-\delta< |s| < |s_0|-\delta\rbrace$. Then for all $s\in U$, letting $\bigtriangleup=(E-\delta, E_+\delta)$:

$$E_\bigtriangleup (H_a(s)) [H_a,iA](s) E_\bigtriangleup (H_a(s))\ge (\tau-\epsilon)  E_\bigtriangleup (H_a(s))$$

since the projections $E_\bigtriangleup (H_a(s))$ are zero for all such $s$ so we may have any constant we wish (in place of $(\tau-\epsilon)$).\bigskip

Now consider $|s_0|\ge \tau$. By Lemma \ref{alpha3lemma1}, we may select $\delta<\delta_0$ and $U$ containing $s_0$ so that for all $s \in U$, letting $\bigtriangleup=(E-\delta,E+\delta)$:

$$E_\bigtriangleup (H_a(s)) [h_a,iA^a] E_\bigtriangleup (H_a(s)) \ge -\epsilon E_\bigtriangleup (H_a(s))$$

Then, we have

\begin{align*}
E_\bigtriangleup (H_a(s)) [H_a,iA](s) E_\bigtriangleup (H_a(s)) &= E_\bigtriangleup (H_a(s)) \Bigg( |s| + [h_a,iA^a] \Bigg) E_\bigtriangleup (H_a(s))\\
&\ge  \tau E_\bigtriangleup (H_a(s)) - \epsilon E_\bigtriangleup (H_a(s))\\
&= (\tau-\epsilon) E_\bigtriangleup (H_a(s))\\
\end{align*}

for all $s\in U$.\bigskip

So, for any fiber $s_0$: there exists a $U$ containing $s_0$ so that for all $s \in U$, letting $\bigtriangleup = (E-\delta, E+\delta)$:

\begin{align*}
E_\bigtriangleup (H_a(s)) [H_a,iA](s) E_\bigtriangleup (H_a(s))
&\ge (\tau-\epsilon) E_\bigtriangleup (H_a(s))\\
&\ge (d(E,a)-2\epsilon) E_\bigtriangleup(H_a(s))\\
\end{align*}

where the last line holds because $\tau \ge d(E,a)-\epsilon$. By a renaming of $\epsilon$ we have our conclusion.

\end{proof}

\begin{proof}[Proof for $E>0$]

Consider $s_0\le \tau$. By Lemma \ref{alpha3lemma2}, we may select $\delta<\delta_0$ and $U$ containing $s_0$ so that for all $s\in U$, letting $\bigtriangleup = (E-\delta, E+\delta)$:

$$E_\bigtriangleup (H_a(s)) (-V_{12}(x)-x\cdot \bigtriangledown V_{12}(x))  E_\bigtriangleup (H_a(s)) \ge -\frac{\epsilon}{2} E_\bigtriangleup (H_a(s))$$

Then, we have

\begin{align*}
E_\bigtriangleup (H_a(s)) [H_a,iA](s) E_\bigtriangleup (H_a(s)) &=E_\bigtriangleup (H_a(s)) \Bigg( 2p^2 + |s| + V_{12}(x) - V_{12}(x) -x\cdot \bigtriangledown V_{12}(x) \Bigg) E_\bigtriangleup (H_a(s))\\
& \ge (E-\delta_0) E_\bigtriangleup (H_a(s)) + E_\bigtriangleup (H_a(s)) \Bigg( - V_{12}(x) - x\cdot \bigtriangledown V_{12}(x) \Bigg) E_\bigtriangleup (H_a(s))\\
&\ge (E-\delta_0 - \frac{\epsilon}{2}) E_\bigtriangleup (H_a(s))\\
&\ge (E-\epsilon)  E_\bigtriangleup (H_a(s))
\end{align*}

for all $s\in U$.

The case  $|s_0|\ge \tau$ for $E>0$ is handled in the same way as the for $E<0$; the same estimate with constant $(\tau-\epsilon)$ holds.

So, for any fiber $s_0$: there exists a $\delta>0$ and a $U$ containing $s_0$ so that for all $s \in U$, letting $\bigtriangleup=(E-\delta, E+\delta)$:

\begin{align*}
E_\bigtriangleup (H_a(s)) [H_a,iA](s) E_\bigtriangleup (H_a(s)) &\ge \min (\tau-\epsilon,E-\epsilon)E_\bigtriangleup (H_a(s))\\
&\ge \min (d(E,a)-2\epsilon,E-\epsilon)E_\bigtriangleup (H_a(s))\\
\end{align*}

By a renaming of $\epsilon$ we have our conclusion.
\end{proof}

Using this, the proof of Lemma \ref{alpha3} is the same as the proof of Lemma \ref{alpha2}.

\subsection{The cluster $(xy)(0)$}\label{(xy)(0)}

In what follows, fix $a = (xy)(0)$, the cluster decomposition corresponding to the electron-photon cluster. We aim to prove the following.

\begin{lemma}\label{alpha4}

Fix $\epsilon>0$ and an energy $E\neq 0$ not an eigenvalue of the subsystem Hamiltonian $h_a = \frac{1}{4}(p^a)^2+\frac{1}{2}|p^a|+V_{23}(x^a)$. Then there exists a $\delta> 0$ so that $H_a$ satisfies a Mourre estimate at $E$ with conjugate operator $A$, width $\delta$, and constant $\alpha_{(xy)(0)}$, where
$$\alpha_{(xy)(0)} = \begin{cases}
2d(E,a)-\epsilon: \text{ } E<0\\
\min(2d(E,a)-\epsilon,E-\epsilon): \text{ } E>0
\end{cases}$$
\end{lemma}

As done for the other 2-cluster decompositions, we can write the commutator $[H_a,iA]$ as a direct integral over the fibered commutators

$$[H,iA](s) = \frac{1}{2}(p^a)^2+\frac{1}{2}p^a\cdot s + \frac{1}{2} \frac{(p^a)^2-s\cdot p^a}{|p^a-s|} - x^a \cdot V_{23}(x^a)$$

We note that $s$ does not separate out from the fibered commutators, which means that the positive commutator estimate in the large $s$ case will require a different strategy than for the other decompositions (The small $s$ case will again exploit the choice of $E$ away from thresholds). This different strategy requires us to know a little more about the spectrum of $H_a(s)$.

\begin{lemma}[The spectrum of $H_{(xy)(0)}(s)$]\label{alpha4spec}
The continuous spectrum of $H_a(s)$ is $[\min_{p^a}(\frac{1}{4}(p^a+s)^2+\frac{1}{2}|p^a-s|), \infty)$, and the eigenvalues of $H_a(s)$ are of the form $\lambda_0+s^2$, where $\lambda_0$ is an eigenvalue of $h_a$.
\end{lemma}

\begin{proof}
The essential spectrum of $H_a(s)$ is contained in $[\min_{p^a}(\frac{1}{4}(p^a+s)^2+\frac{1}{2}|p^a-s|), \infty)$ by Weyl's theorem. This minimum can be computed as

$$\min_{p^a}(\frac{1}{4}(p^a+s)^2+\frac{1}{2}|p^a-s|) = \begin{cases} s^2 : \hspace{.5cm}|s| \le \frac{1}{2}\\ |s|-\frac{1}{4} :\hspace{.5cm} |s| > \frac{1}{2}\\
\end{cases}$$

By the assumption (\ref{SPEC}), the essential spectrum is exactly the absolutely continuous spectrum; there are no eigenvalues of $H_a(s)$ in the continuous spectrum region $[\min_{p^a}(\frac{1}{4}(p^a+s)^2+\frac{1}{2}|p^a-s|), \infty)$. We concern ourselves now with the eigenvalues. \bigskip

The operator $H_a(s)$ is unitarily equivalent to the operator $B(s):= \frac{1}{4}(p^a+2s)^2+\frac{1}{2}|p^a|+V_{23}(x^a)$. The unitary equivalence is given by $U(s):=e^{-isp^a}$, so that $U(-s)H_a(s)U(s)= B(s)$. Thus $H_a(s)$ and $B(s)$ have the same spectrum for any choice of $s$. \bigskip

Pick a unit vector $v \in \mathbb{R}^3$. We consider the family of operators $B(tv)$ for $t\in \mathbb{R}$. These form a self-adjoint holomorphic family of operators in the sense of Kato. \bigskip

Specifically, fix a domain $D_0 \in \mathbb{C}$ symmetric with respect to the real axis. We have that $B(tv)$ is a closed, densely defined operator for all $t\in D_0$. As a function of $t$, $B(tv)$ is holomorphic for $t\in D_0$, and $B(\bar{t}v)=B(tv)^*$. We know it is holomorphic because we can compute its derivatives (either in the weak sense or strong sense):

$$\frac{dB(tv)}{dt}=2t+p^a\cdot v$$

Because this is a holomorphic family of operators, its isolated eigenvalues and their eigenvectors can be thought of as varying holomorphically in a certain sense. Fix an isolated eigenvalue $\lambda_0$ of $B(0)$. We know from the Mourre estimate that all eigenvalues below $0$ are simple and do not accumulate, so we can draw a curve $\Gamma$ around $\lambda_0$ that is entirely contained in the resolvent set and encloses no other points of the spectrum of $B(0)$. It is then known that for small $t$, all eigenvalues of $B(tv)$ inside $\Gamma$ can be described by a function $\lambda(t)$ that is analytic in a region about $t=0$, satisfies $\lambda(0)=\lambda_0$, and gives a real eigenvalue of $H_a(tv)$ for each real $t$ in its domain- also there are no other eigenvalues of $B(tv)$ for any small enough $t$ in the interior of $\Gamma$. Even better, there exists at least one analytic family of real-valued eigenvectors $\psi_t(p^a)$ for $\lambda(t)$, as long as this family $\lambda(t)$ continues to exist.\bigskip

We want to compute $\lambda(t)$. What follows is an application of the Feynman-Hellman theorem.

\begin{align*}
\frac{d\lambda(t)}{dt} &= \frac{d}{dt}\langle \psi_t, B(tv) \psi_t\rangle \\
&= \langle \psi_{t}, \frac{dB(tv)}{dt} \psi_{t}\rangle + \langle \frac{d\psi_t}{dt}, B(tv) \psi_t \rangle + \langle \psi_t, B(tv)  \frac{d\psi_t}{dt} \rangle \\
&= \langle \psi_{t}, \frac{dB(tv)}{dt} \psi_{t}\rangle + \lambda(t) \frac{d}{dt}\langle \psi_t, \psi_t\rangle\\
&= \langle \psi_{t}, \frac{dB(tv)}{dt} \psi_{t}\rangle \\
&= \langle \psi_{t}, (2t+p^a \cdot v) \psi_{t}\rangle \\
&= 8t\\
\end{align*}

where the last step is justified by integration by parts and the fact that $\psi_t$ could be taken to be real-valued. Integrating, we obtain

$$\lambda(t)= t^2+ \lambda_0$$

While the function $\lambda(t)$ may not exist for all $t$, this process can be analytically continued as long as $\lambda(t)$ remains below the continuous spectrum of $B(tv)$. Here is why: suppose that $\lambda(t_0)$ is below the continuous spectrum of $B(tv)$, and we can define an operator $B((t-t_0)v)$ and compute its derivatives in the same way. So the spectrum of $H_a(tv)$ below the continuous spectrum consists only of eigenvalues of the form $t^2 + \lambda_0$ for isolated eigenvalues $\lambda_0$ of $H_a(0)$.\bigskip
\end{proof}

At this point we have everything we need to handle the $E<0$ case.

\begin{proof}[Proof of Lemma \ref{alpha4} for $E<0$]\bigskip

Fix $\epsilon$ and $E<0$ as in the lemma. We may select $\delta$ so small that $(E-\delta, E+\delta)$ contains no eigenvalues of $h_a$, and also so that $\delta\le \frac{\epsilon}{2}$. Fix $\tau$ so small that $\bigtriangleup = (E-\delta, E+\delta)$ contains no eigenvalues of $H_a(s)$ for $|s| \le \tau$. From what we know of the spectrum of $H_a(s)$, $\tau$ may be selected so that $\tau^2 \ge d(E,a)-\epsilon$. \bigskip

Consider the small-momentum case, $|s|\le \tau$. We have $E_\bigtriangleup(H_a(s)) = 0$ for all such $s$,

$$E_\bigtriangleup(H_a(s)) [H_a,iA](s) E_\bigtriangleup(H_a(s)) \ge (2\tau^2) E_\bigtriangleup(H_a(s))$$

Since the projections were $0$, we could have put any constant where $2\tau^2$ is.\bigskip

Now consider the large momentum case, $|s|\ge \tau$. We know that for all such $s$, $E_\bigtriangleup(H_a(s))$ is a (possibly zero) projection onto a finite-dimensional subspace of the pure point spectrum of $H_a(s)$. Because of this, we can make use of the virial theorem (specifically, that $[H_a(s), iA^a]=0$ on eigenvectors of $H_a(s)$).

\begin{align*}
E_\bigtriangleup (H_a(s))[H_a,iA](s) E_\bigtriangleup(H_a(s)) &= E_\bigtriangleup (H_a(s))[H_a,iA](s) E_\bigtriangleup(H_a(s))\\
&= E_\bigtriangleup (H_a(s))\left( [H_a,iA](s) - [H_a(s),iA^a]\right) E_\bigtriangleup (H_a(s))\\
&= E_\bigtriangleup (H_a(s))\left( \frac{1}{2} s^2 + \frac{1}{2} s \cdot p^a + \frac{1}{2} \frac{s^2 - s \cdot p^a}{|s-p^a|} \right) E_\bigtriangleup (H_a(s)) \\
&= E_\bigtriangleup (H_a(s))\left( s\cdot\bigtriangledown_s H_a(s)\right)E_\bigtriangleup (H_a(s))\\
&= 2s^2 E_\bigtriangleup (H_a(s))\\
&\ge 2\tau^2 E_\bigtriangleup (H_a(s))\\
\end{align*}

where second-to-last last equality comes from the Feynman-Hellman theorem. Finally, we have\bigskip

\begin{align*}
E_\bigtriangleup (H_a(s)) [H_a,iA](s) E_\bigtriangleup (H_a(s)) &\ge 2\tau^2 E_\bigtriangleup (H_a(s))\\
& \ge 2(d(E,a)-\epsilon) E_\bigtriangleup (H_a(s))\\
\end{align*}

Then the following is immediate:

\begin{align*}
E_\Delta(H_a) [H_a,iA] E_\Delta(H_a) &= \int_{s\in \mathbb{R}^3}^\oplus E_{\Delta}(H_a(s)) [H_a,iA](s) E_{\Delta}(H_a(s)) ds\\
&\ge \int_{s\in \mathbb{R}^3}^\oplus 2(d(E,a)-\epsilon) E_{\bigtriangleup}(H_a(s)) ds\\
&= 2(d(E,a)-\epsilon) E_{\bigtriangleup}(H_a)\\
\end{align*}

By a renaming of $\epsilon$ we have the conclusion.\bigskip
\end{proof}

To attack the $E>0$ case, we will need to be more judicious in our selection of width $\delta$ for different fibers $s$. As a result, we need to make a covering argument over $\mathbb{R}^3_s$. First we prove  what we need for individual $s$.

\begin{lemma}\label{fixsfinddelta}
Fix $\epsilon>0$, an energy $E > 0$ (which is not an eigenvalue of the subsystem Hamiltonian $h_a$), and a choice of $s\in \mathbb{R}^3$. Then there there exists $\delta=\delta(s)>0$ so that, taking $\bigtriangleup=(E-\delta,E+\delta)$, we have:

$$E_\bigtriangleup(H_a(s)) [H_a,iA](s) E_\bigtriangleup(H_a(s)) \ge \min(d(E,a)-\epsilon, E-\epsilon) E_\bigtriangleup(H_a(s))$$
\end{lemma}

\begin{proof}
Fix $\epsilon$ and $E>0$ as in the lemma. We may select $\delta_0$ so small that $(E-\delta_0, E+\delta_0)$ contains no eigenvalues of $h_a$, and also that $\delta_0< \frac{\epsilon}{2}$. Fix $\tau$ so small that $\bigtriangleup_0 = (E-\delta_0, E+\delta_0)$ contains no eigenvalues of $|H_a(s)|$ for $|s| \le \tau$. From what we know of the spectrum of $H_a(s)$, $\tau$ may be selected so that $\tau^2 \ge d(E,a)-\epsilon$. \bigskip

Consider the small-momentum case, $|s|\le \tau$. On these fibers:

\begin{align*}
E_{\bigtriangleup_0}(H_a(s)) [H_a,iA](s) E_{\bigtriangleup_0}(H_a(s)) &= E_{\bigtriangleup_0}(H_a(s)) \left( \frac{1}{2}(p^a+s)^2 + H_a(s) - V_{23}(x^a)-x^a \cdot \bigtriangledown V_{23}(x^a) \right) E_{\bigtriangleup_0}(H_a(s))\\
&\ge (E-\delta_0)E_{\bigtriangleup_0}(H_a(s)) + E_{\bigtriangleup_0}(H_a(s)) K E_{\bigtriangleup_0}(H_a(s))
\end{align*}

by the functional calculus, where $K= - V_{23}(x^a)-x^a \cdot \bigtriangledown V_{23}(x^a)$ is relatively $H_a(s)$-compact. Since $\bigtriangleup_0$ contains no eigenvalues of $H_a(s)$, $E_\bigtriangleup(H_a(s))\rightarrow 0$ in the strong operator topology as $\bigtriangleup\searrow 0$. Therefore $E_\bigtriangleup(H_a(s)) K E_\bigtriangleup(H_a(s)) \rightarrow 0$ in norm as $\bigtriangleup \searrow 0$. So, by choosing $\delta_1\le \delta_0$ small enough, and letting $\bigtriangleup_1 = (E-\delta_1,E+\delta_1)$, we can ensure that

\begin{align*}
E_{\bigtriangleup_1}(H_a(s)) [H_a,iA](s) E_{\bigtriangleup_1}(H_a(s)) &\ge  (E-\delta_0)E_{\bigtriangleup_1}(H_a(s)) -\frac{\epsilon}{2}E_{\bigtriangleup_1}(H_a(s))\\
\end{align*}

and therefore that

\begin{align*}
E_{\bigtriangleup_1}(H_a(s)) [H_a,iA](s) E_{\bigtriangleup_1}(H_a(s)) &\ge  (E-\epsilon)E_{\bigtriangleup_1}(H_a(s)) \\
\end{align*}

This is all we need to do in our analysis of fibers $|s|\le \tau$.\bigskip

Now we consider the large-momentum case, $|s|\ge \tau$. We need to be very careful in selecting the width, to satisfy a whole host of auxiliary inequalities. Let $\delta_2<\delta_0$ be so small that, letting $\bigtriangleup_2 = (E-\delta_2,E+\delta_2)$:

\begin{equation}\label{last1}
\Vert \Big( E_{\bigtriangleup_2}(H_a(s)) -E_{\bigtriangleup_2 pp} (H_a(s))\Big) K \Big(E_{\bigtriangleup_2}(H_a(s)) -E_{\bigtriangleup_2 pp} (H_a(s))\Big) \Vert \le \frac{\epsilon}{4}
\end{equation}

where $K= -x^a-x^a \cdot \bigtriangledown V_{23}(x^a)$ is relatively $H_a(s)$-compact. \bigskip

Because $E_{\bigtriangleup_2}(H_a(s)) K E_{\bigtriangleup_2}(H_a(s))$ is compact, we may select a finite-dimensional projection $F$ so that

\begin{equation}\label{last2}
\Vert \Big( E_{\bigtriangleup_2}(H_a(s))-E_{\bigtriangleup_2 pp} (H_a(s))\Big) K\Big( E_{\bigtriangleup_2}(H_a(s))-E_{\bigtriangleup_2 pp} (H_a(s))\Big) - \Big( E_{\bigtriangleup_2}(H_a(s))-F\Big) K\Big( E_{\bigtriangleup_2}(H_a(s))-F\Big) \Vert \le \frac{\epsilon}{4}
\end{equation}

Now define $C:=F[H_a,iA](s)(E_{\bigtriangleup_2}(H_a(s))-E_{\bigtriangleup pp} (H_a(s)))$. Evidently $K_1= -\epsilon^{-1}C^*C$ is a compact operator. If we select $\delta_3<\delta_2$ small enough small enough, then letting $\bigtriangleup_3=(E-\delta_3,E+\delta_3)$, we have:

\begin{equation}\label{last3}
\Vert \Big( E_{\bigtriangleup_3}(H_a(s)) -E_{\bigtriangleup_3 pp} (H_a(s))\Big) K_1 \Big( E_{\bigtriangleup_3}(H_a(s)) -E_{\bigtriangleup_3 pp} (H_a(s))\Big) \Vert \le \epsilon
\end{equation}

From the same argument as the $E<0$ case, we have that

\begin{equation}\label{pppart}
E_{\bigtriangleup_0 pp}(H_a(s)) [H_a,iA](s) E_{\bigtriangleup_0 pp}(H_a(s)) \ge 2 \tau^2 E_{\bigtriangleup_0 pp}(H_a(s))
\end{equation}

Moreover, we can compute

\begin{equation}\label{last4}
\begin{split}
E_{\bigtriangleup_2}(H_a(s)) [H_a,iA](s) E_{\bigtriangleup_2}(H_a(s)) &= E_{\bigtriangleup_2}(H_a(s)) \Bigg( \frac{1}{2}(p^a+s)^2 + H_a(s) + K \Bigg) E_{\bigtriangleup_2}(H_a(s))\\
&\ge (E-\delta_0) E_{\bigtriangleup_2}(H_a(s)) + E_{\bigtriangleup_2}(H_a(s)) K E_{\bigtriangleup_2}(H_a(s))\\
\end{split}
\end{equation}

where $K= - V_{23}(x^a)-x^a \cdot \bigtriangledown_s V_{23}(x^a)$ is $H_a(s)$-compact. We can multiply both sides of (\ref{last4}) on the left and right by $(1- F)$ to obtain

\begin{align*}
&(E_{\bigtriangleup_2}(H_a(s))-F)[H,iA](s)(E_{\bigtriangleup_2}(H_a(s))-F) \\
&\ge (E-\delta_0)(E_{\bigtriangleup_2}(H_a(s))-F) + (E_{\bigtriangleup_2}(H_a(s))-F) K (E_{\bigtriangleup_2}(H_a(s))-F)\\
\end{align*}

Then we use (\ref{last1}) and (\ref{last2}):\bigskip

$$(E_{\bigtriangleup_2}(H_a(s))-F)[H,iA](s)(E_{\bigtriangleup_2}(H_a(s))-F) $$
\begin{equation}\label{last5}
\ge (E-\delta_0)(E_{\bigtriangleup_2}(H_a(s))-F) - \frac{\epsilon}{2}
\end{equation}

Additionally:

\begin{align*}
&2F[H_a,iA](s) F - E_{\bigtriangleup_2 pp}(H_a(s))[H_a,iA](s) F - F[H_a,iA](s) E_{\bigtriangleup_2 pp}(H_a(s))\\
 &= (E_{\bigtriangleup_2 pp}(H_a(s))- F) [H_a,iA](s) (E_{\bigtriangleup_2 pp}(H_a(s))-F) + F[H_a,iA](s)F - E_{\bigtriangleup_2 pp}(H_a(s))[H_a,iA](s)E_{\bigtriangleup_2 pp}(H_a(s))\\
&= 2s^2(E_{\bigtriangleup_2 pp}(H_a(s))-F) + 2s^2 F -2s^2 E_{\bigtriangleup_2 pp}(H_a(s))\\
&= 0\\
\end{align*}

Therefore

\begin{equation}\label{last6}
\begin{split}
&(E_{\bigtriangleup_2}(H_a(s))- F)[H_a,iA](s) F + F[H_a,iA](s)(E_{\bigtriangleup_2}(H_a(s))- F) \\
&= \Big( E_{\bigtriangleup_2}(H_a(s))- E_{\bigtriangleup_2 pp}(H_a(s))\Big) [H_a,iA](s) F + F[H_a,iA](s)\Big( E_{\bigtriangleup_2}(H_a(s))- E_{\bigtriangleup_2 pp}(H_a(s))\Big)\\
\end{split}
\end{equation}

since the difference between the left hand side and the right hand side was just calculated to be zero.\bigskip

From the inequality $(\epsilon^{-1/2} C + \epsilon^{1/2} F)^*(\epsilon^{-1/2} C + \epsilon^{1/2} F) \ge 0$ we obtain

\begin{align*}
&\Big( E_{\bigtriangleup_2}(H_a(s))- E_{\bigtriangleup_2 pp}(H_a(s))\Big)[H_a,iA](s) F + F[H_a,iA](s)\Big( E_{\bigtriangleup_2}(H_a(s))- E_{\bigtriangleup_2 pp}(H_a(s))\Big)\\
&\ge -\epsilon F + \Big( E_{\bigtriangleup_2}(H_a(s))- E_{\bigtriangleup_2 pp}(H_a(s))\Big) K_1 \Big( E_{\bigtriangleup_2}(H_a(s))- E_{\bigtriangleup_2 pp}(H_a(s))\Big)\\
\end{align*}

 but then, applying (\ref{last6}) to this, we can substitute out the left-hand side:

 \begin{equation}\label{last7}
 \begin{split}
 &(E_{\bigtriangleup_2}(H_a(s))- F)[H_a,iA](s) F + F[H_a,iA](s)(E_{\bigtriangleup_2}(H_a(s))- F)\\
&\ge -\epsilon F + \Big( E_{\bigtriangleup_2}(H_a(s))- E_{\bigtriangleup_2 pp}(H_a(s))\Big) K_1 \Big( E_{\bigtriangleup_2}(H_a(s))- E_{\bigtriangleup_2 pp}(H_a(s))\Big)\\
 \end{split}
 \end{equation}

We are ready to tackle the main estimate. We have

\begin{align*}
E_{\bigtriangleup_2}(H_a(s))[H_a,iA](s)E_{\bigtriangleup_2}(H_a(s)) &= F [H_a,iA](s) F  \\
&\hspace{.5cm}+ (E_{\bigtriangleup_2}(H_a(s)) -F) [H_a,iA](s) F\\
&\hspace{.5cm}+  F[H_a,iA](s) (E_{\bigtriangleup_2}(H_a(s)) -F)\\
&\hspace{.5cm}+  (E_{\bigtriangleup_2}(H_a(s)) -F)[H_a,iA](s) (E_{\bigtriangleup_2}(H_a(s)) -F)\\
\end{align*}

Applying (\ref{pppart}) to the first term, (\ref{last5}) to the last term, and (\ref{last7}) to the middle two terms:

\begin{align*}
&\ge 2\tau^2 F + (E-\delta_0)(E_{\bigtriangleup_2}(H_a(s))-F) -\frac{\epsilon}{2} -\epsilon F \\
&\hspace{1cm} + \Big( E_{\bigtriangleup_2}(H_a(s))- E_{\bigtriangleup_2 pp}(H_a(s))\Big) K_1 \Big( E_{\bigtriangleup_2}(H_a(s))- E_{\bigtriangleup_2 pp}(H_a(s))\Big) \\
&\ge \min(2\tau^2-\epsilon, E-\delta_0) E_{\bigtriangleup_2}(H_a(s))-\frac{\epsilon}{2} + \Big( E_{\bigtriangleup_2}(H_a(s))- E_{\bigtriangleup_2 pp}(H_a(s))\Big) K_1 \Big( E_{\bigtriangleup_2}(H_a(s))- E_{\bigtriangleup_2 pp}(H_a(s))\Big)\\
\end{align*}

Multiplying on both sides by $E_{\bigtriangleup_3 pp}(H_a(s))$ and using (\ref{last3}):

$$\ge \Big(\min(2\tau^2-\epsilon, E-\epsilon)-\epsilon\Big)  E_{\bigtriangleup_3 pp}(H_a(s))$$

Since $\tau^2 \ge d(E,a)-\epsilon$, the conclusion follows by a renaming of $\epsilon$.

\end{proof}

Building on the inequality for individual fibers $s$, we can find a single width that works for all fibers, thus finishing the analysis for this cluster.

\begin{proof}[Proof of Lemma \ref{alpha4} for $E>0$]\bigskip

Fix $\epsilon$ and $E>0$ as in the lemma. Fix any $s_0 \in \mathbb{R}^3$. By Lemma \ref{fixsfinddelta}, there exists $\delta_0$ so that letting $\bigtriangleup_0 = (E-\delta-0, E+\delta-0)$:

$$E_{\bigtriangleup_0}(H_a(s_0))[H_a,iA](s_0)E_{\bigtriangleup_0}(H_a(s_0)) \ge \min(d(E,a)-\epsilon, E-\epsilon)E_{\bigtriangleup_0}(H_a(s_0))$$

Let $f\in C^\infty_0$ have support contained in $\bigtriangleup_0$. Since $H_a(s)$ is a holomorphic family of operators, it is continuous in $s$ in the norm resolvent sense; so $f(H_a(s))$ is operator norm continuous in $s$. The next claim is that $f(H_a(s))[H_a,iA](s)f(H_a(s))$ is also operator norm continuous in $s$. Fix $\epsilon_0>0$. We compute:

\begin{align*}
&\Vert f(H_a(s)) [H_a,iA](s) f(H_a(s)) - f(H_a(s_0)) [H_a,iA](s_0) f(H_a(s_0)) \Vert \\
& \le \Vert \big( f(H_a(s)) - f(H_a(s_0))\big) [H_a,iA](s) f(H_a(s_0)) \Vert \\
& \hspace{1cm} + \Vert f(H_a(s)) [H_a,iA](s) \big( f(H_a(s)) - f(H_a(s_0))\big) \Vert \\
& \hspace{1cm} + \Vert f(H_a(s_0)) \big( [H_a,iA](s) -[H_a,iA](s_0) \big) f(H_a(s_0))\Vert \\
\end{align*}

Regarding $\Vert f(H_a(s_0)) \big( [H_a,iA](s) -[H_a,iA](s_0) \big) f(H_a(s_0))\Vert$: Since

\begin{align*}
[H_a,iA](s) -[H_a,iA](s_0) &= p^a \cdot (s-s_0) + \frac{1}{2}(s^2-s_0^2) + \frac{1}{2}(|p^a-s|-|p^a-s_0|)\\
\end{align*}

and

$$\frac{1}{\frac{1}{4}(p^a+s_0)^2 + \frac{1}{2}|p^a-s_0|+i}\Bigg(p^a \cdot (s-s_0) + \frac{1}{2}(s^2-s_0^2) + \frac{1}{2}(|p^a-s|-|p^a-s_0|)\Bigg)$$

goes to $0$ in the $L^\infty$ norm as $s\rightarrow s_0$, we need only check that $f(H_a(s_0)) \big( \frac{1}{4}(p^a+s_0)^2 + \frac{1}{2}|p^a-s_0|+i \big)$ is uniformly bounded in $s$, which it evidently is (this requires (\ref{RC1})). Therefore if $|s-s_0|$ is small enough, $\Vert f(H_a(s_0)) \big( [H_a,iA](s) -[H_a,iA](s_0) \big) f(H_a(s_0))\Vert \le \epsilon_0/3$. \bigskip

Considering $\Vert \big( f(H_a(s)) - f(H_a(s_0))\big) [H_a,iA](s) f(H_a(s_0)) \Vert$, we need only check that $[H_a,iA](s)f(H_a(s_0))$ is uniformly bounded in $s$ as $s$ varies over a small ball $B$ around $s_0$. Then since $f(H_a(s)) \rightarrow f(H_a(s_0))$ in norm as $s\rightarrow s_0$, by choosing $s\in B$ so that $|s-s_0|$ is small enough, we would have that $\Vert \big( f(H_a(s)) - f(H_a(s_0))\big) [H_a,iA](s) f(H_a(s_0)) \Vert \le \epsilon_0/3$.\bigskip

Because we can write

$$[H_a,iA](s) = [H_a,iA](s) -[H_a,iA](s_0) + [H_a,iA](s_0)$$

we can use the fact that $[H_a,iA](s_0) f(H_a(s_0))$ is uniformly bounded in $s$ to realize all we need to do is show that $\big([H_a,iA](s) -[H_a,iA](s_0)\big) f(H_a(s_0))$ is uniformly bounded as $s$ varies over a small ball $B$. But we have already shown that $\big([H_a,iA](s) -[H_a,iA](s_0)\big) f(H_a(s_0))$ goes to $0$ in norm as $s\rightarrow s_0$.\bigskip

Finally, consider $\Vert f(H_a(s)) [H_a,iA](s) \big( f(H_a(s)) - f(H_a(s_0))\big) \Vert $. We need to show that $ f(H_a(s)) [H_a,iA](s)$ is uniformly bounded in $s$ in a small ball $B$ around $s_0$; then it follows that by taking $s\in B$ so that $|s-s_0|$ is small enough, we would have $\Vert f(H_a(s)) [H_a,iA](s) \big( f(H_a(s)) - f(H_a(s_0))\big) \Vert \le \epsilon_0/3$. \bigskip

Because can write

$$f(H_a(s)) = f(H_a(s)) \big( H_a(s) + i\big) \frac{1}{H_a(s)+i}\big( \frac{1}{4}(p^a+s_0)^2+\frac{1}{2}|p^a-s_0|+i \big) \frac{1}{\frac{1}{4}(p^a+s_0)^2+\frac{1}{2}|p^a-s_0|+i }$$

and $\frac{1}{\frac{1}{4}(p^a+s_0)^2+\frac{1}{2}|p^a-s_0|+i }[H_a,iA](s)$ is bounded as before, we need only check that $ \frac{1}{H_a(s)+i}\big( \frac{1}{4}(p^a+s_0)^2+\frac{1}{2}|p^a-s_0|+i \big)$ is uniformly bounded in $s$. This follows from (\ref{RC1}).\bigskip

Therefore $f(H_a(s))[H_a,iA](s)f(H_a(s))$ is norm continuous in $s$. Since this works for any $f\in C^\infty$ supported in $\bigtriangleup_0$, we fix such an $f$ that is equal to $1$ on a smaller interval $(E-\delta, E+\delta) = \bigtriangleup \subset \bigtriangleup_0$. Let $U$ be an open set containing $s_0$ so that for all $s\in U$, we have

$$\Vert f(H_a(s)) [H_a,iA](s) f(H_a(s)) - f(H_a(s_0)) [H_a,iA](s_0) f(H_a(s_0)) \Vert\le epsilon$$

and $\Vert f(H_a(s_0))-f(H_a(s))\Vert \le \epsilon$

Then, since

$$E_{\bigtriangleup_0}(H_a(s_0))[H_a,iA](s_0)E_{\bigtriangleup_0}(H_a(s_0)) \ge \min(d(E,a)-\epsilon, E-\epsilon)E_{\bigtriangleup_0}(H_a(s_0))$$

we have that, by multiplying:

$$f(H_a(s_0))[H_a,iA](s_0)f(H_a(s_0)) \ge \min(d(E,a)-\epsilon, E-\epsilon)f(H_a(s_0))$$

And then for all $s\in U$, we must have

$$f(H_a(s))[H_a,iA](s)f(H_a(s)) \ge \min(d(E,a)-\epsilon, E-\epsilon)f(H_a(s))-2\epsilon $$

Finally, multiplying through by $E_\bigtriangleup(H_a(s))$:

$$E_{\bigtriangleup}(H_a(s))[H_a,iA](s)E_{\bigtriangleup}(H_a(s)) \ge (\min(d(E,a)-\epsilon, E-\epsilon)-2\epsilon) E_{\bigtriangleup}(H_a(s))$$

Therefore given any $\epsilon$, an energy $E>0$ not an eigenvalue of the subsystem Hamiltonian $h_a$, and a value of $s_0\neq 0$, there exists $\delta >0$ and an open set $U\in \mathbb{R}^3$ containing $s_0$ so that for all $s\in U$, letting $\bigtriangleup = (E-\delta,E+\delta)$, the above Mourre estimate holds. Thus the same covering argument as in the proof of Lemma \ref{alpha2} works, and by a renaming of $\epsilon$ we have our conclusion.

\end{proof}

\subsection{Completing the Mourre estimate}

\begin{proof}
Given a nonzero, nonthreshold energy $E$ we may select a single $\delta$ small enough to invoke lemmas \ref{alpha1}, \ref{alpha2}, \ref{alpha3}, and \ref{alpha4}. We can select $\epsilon$ so small that all the constants $\alpha_a$ are positive. Then we can select a $C^\infty_0$ function $f$ that is $0$ outside of $(E-\delta,E+\delta)$ and is equal to $1$ on a smaller interval $\bigtriangleup$ containing $E$. We employ the localization Lemma \ref{ims} using this $f$:

\begin{align*}
\sum_a f(H)[H,iA]f(H)&= \text{(compact operators)} + \sum_{a\neq (xy0)} j_a f(H_a)[H_a,iA]f(H_a)j_a \\
&\ge \sum_{a\neq (xy0)} \alpha_a j_a f(H_a) j_a + \text{(compact operators)}\\
&\ge (\min_{a\neq (xy0)} \alpha_a) f(H)^2 + \text{(compact operators)}\\
\end{align*}

This last step is by e.g. Lemma 4.20 in \cite{cfks}. Multiplying on both sides by $E_\bigtriangleup(H)$ concludes the proof of Theorem \ref{mourre}.
\end{proof}

In order to invoke Mourre's result, we check that $H$ and $A$ also satisfy (\ref{2COMM}). Since $D(C)=D(H)$, a computation using Lemma \ref{commlemma} reveals that $[C,iA]$ extends to
 $$4p^2 + |k| + x\cdot \bigtriangledown (x\cdot \bigtriangledown V_{12}(x)) + y\cdot \bigtriangledown (y\cdot \bigtriangledown V_{13})   + V_{23}(x-y) \cdot \bigtriangledown ((x-y)\cdot \bigtriangledown V_{23}(x-y))$$,
 which is a bounded operator on $D(H)$ by Kato-Rellich and the potential assumptions (\ref{RB2}). Therefore by \cite{mourre} or Theorem 1.1 in \cite{pss}, we have that the point spectrum of $H$ consists of simple eigenvalues which only may accumulate at thresholds, and that there is no singular continuous spectrum.\bigskip

\section{Local decay and minimal velocity estimates}

\subsection{Local decay}

A major consequence of the Mourre estimate is local decay estimates. An abstract result originally due to Mourre can be found in (\cite{pss}, Thm 7.8), which we recreate here.\bigskip

\begin{lemma}\label{PSS}
Suppose that $H$, $H_0$, and $A$ are three self-adjoint operators so that $D(H)=D(H_0)$ and $H$ and $H_0$ are both bounded from below. Assume that hypotheses (\ref{formal1})-(\ref{formal4}) and (\ref{2COMM}) hold for $H$ and $A$. Moreover, assume that (\ref{formal1})-(\ref{formal4}) hold for $H_0$ and $A$ so that $[H_0,iA]$ extends to an operator defined on $D(H)$. Finally, assume that the core of test vectors $S$ used to define the operator $[H_0,iA]$ is mapped into itself by $A$. Then, let $\bigtriangleup$ be an interval in which a Mourre estimate holds for $H$ with conjugate operator $A$, so that $\bigtriangleup$ does not contain any eigenvalues of $A$. We have

$$\sup_{0< \epsilon < 1} \Vert (|A|+1)^{-\mu} (H-\lambda-i\epsilon)^{-1} (|A|+1)^{-\mu} \Vert<\infty$$

for any fixed $\mu>\frac{1}{2}$, where this holds uniformly as $\lambda$ runs through compact subsets of $\bigtriangleup$.
\end{lemma}

Since these extra conditions hold under our assumptions, we are able to invoke this lemma for our $H$, $H_0$, and $A$. Next, this estimate can be modified to remove the reference to the operator $A$ and instead say something about the position $X$. Define the notation $\langle X \rangle := \sqrt{X^2+1}$.  Specifically, we want to prove that for any interval $\bigtriangleup$ where the Mourre estimate holds for $H$,

\begin{equation}\label{swapgoal}
\sup_{0< \epsilon < 1} \Vert \langle X \rangle^{-\mu} (H-\lambda-i\epsilon)^{-1} \langle X \rangle^{-\mu} \Vert<\infty
\end{equation}

for any fixed $\mu>\frac{1}{2}$, where this holds uniformly as $\lambda$ runs through compact subsets of $\bigtriangleup$. The main fact used to perform this swap is that

\begin{equation}\label{swapfact}
(|A|+1)^\mu (H+i)^{-1} \langle X \rangle^{-\mu}
\end{equation}

is a bounded operator for any $0 \le \mu \le 1$. Assuming (\ref{swapfact}) is indeed bounded for any $0 \le \mu \le 1$, we proceed to prove (\ref{swapgoal}). Let $L^2_\mu$ be the weighted $L^2$ space $\lbrace f \in L^2(\mathbb{R}^6) : \langle X \rangle^{\mu} f \in L^2(\mathbb{R}^6)\rbrace$. Then (\ref{swapgoal}) is equivalent to saying that $(H-\lambda-i\epsilon)^{-1}$ is bounded from $L^2_\mu$ to $L^2_{-\mu}$ uniformly as $\lambda$ and $\epsilon$ vary over the required sets. Since

$$(H-\lambda-i\epsilon)^{-1} = (H+i)^{-1} + (\xi+i)^{-1}(H+i)^{-2}+(\xi+i)^2 (H+i)^{-1} (H-\lambda-i\epsilon)^{-1} (H+i)^{-1}$$

(where $\xi = \lambda + i \epsilon$) it remains to show that $(H+i)^{-1} (H-\lambda-i\epsilon)^{-1} (H+i)^{-1}$ is bounded from $L^2_\mu$ to $L^2_{-\mu}$ uniformly as $\lambda$ and $\epsilon$ vary over the required sets. But for this we can rewrite

\begin{align*}
&\langle X \rangle^{-\mu}(H+i)^{-1} (H-\lambda-i\epsilon)^{-1} (H+i)^{-1}\langle X \rangle^{-\mu} = \\
& \Big( (|A|+1)^\mu (H+i)^{-1} \langle X \rangle^{-\mu}\Big)^* \Big(  (|A|+1)^{-\mu} (H-\lambda-i\epsilon)^{-1} (|A|+1)^{-\mu} \Big) \Big( (|A|+1)^\mu (H+i)^{-1} \langle X \rangle^{-\mu}\Big)\\
\end{align*}

and then using (\ref{swapfact}) and Lemma \ref{PSS} gives us  (\ref{swapgoal}). It remains to prove that (\ref{swapfact}) is bounded for any $\mu > \frac{1}{2}$. Without loss of generality we may assume also that $\mu\le 1$. In fact, we will prove that (\ref{swapfact}) is bounded for $\mu =0$ and for $\mu = 1$, and then use Stein's interpolation theorem for analytic families of operators in order to draw the conclusion. The case $\mu=0$ is evident. We consider the case $\mu=1$. We need only bound

$$(P \cdot X) (H+i)^{-1} \langle X \rangle^{-1}$$

The equalities that follow come from restricting our attention to the dense domain of Schwartz functions.

$$(P \cdot X) (H+i)^{-1} \langle X \rangle^{-1} = S_1 + S_2$$

where

$$S_1 = (P (H+i)^{-1}) \cdot (X \langle X \rangle^{-1})$$

which is bounded, and

\begin{align*}
S_2 &= P \cdot [X, (H+i)^{-1}]\langle X \rangle^{-1} \\
&= P \cdot \Big( (H+i)^{-1} [H_0,X](H+i)^{-1} \Big)\langle X \rangle^{-1}\\
&= P \cdot \Big( (H+i)^{-1} (-2 i p, -2i\frac{k}{|k|} ) (H+i)^{-1} \Big)\langle X \rangle^{-1}\\
\end{align*}

which is also bounded (thinking of $(-2 i p, -2i\frac{k}{|k|} )$ as a vector in $\mathbb{C}^6$, so the dot product makes sense). Since (\ref{swapfact}) is then shown to be bounded for all required $\mu$, we have the desired estimate (\ref{swapgoal}). \bigskip

An operator $B$ on $L^2(\mathbb{R}^6)$ is said to be $H$-smooth if for all $\phi\in L^2(\mathbb{R}^6)$, we have $e^{-itH}\in D(B)$ a.e. $t$ and

\begin{equation}\label{hsmooth}
\int_{-\infty}^\infty \Vert B e^{-itH} \phi \Vert^2 dt \lesssim \Vert \phi \Vert ^2
\end{equation}

We say $B$ is $H$-smooth on $\overline{\Omega}$ if $B E_{\overline{\Omega}}(H)$ is $H$-smooth.\bigskip

This can be interpreted as the observable $B$ decaying along the flow. By the general theory (\cite{rs}, Theorems XIII.25 and XIII.30), the estimate (\ref{swapgoal}) implies that for any interval $\Omega$ not containing eigenvalues or thresholds of $H$, $\langle X \rangle^{-\mu}$ is $H$-smooth on $\overline{\Omega}$ for any $\mu > \frac{1}{2}$. This fact is `local decay', an important tool for proving the existence of wave operators.\bigskip

\subsection{Minimal velocity estimates}

Another important consequence of the Mourre estimate is minimal velocity estimates, which we use in what follows. In all that follows, we write $F(S)$ to signify a smoothed characteristic function of the set defined by $S$ in configuration space. It is known (cf. \cite{ss2}) that the Mourre estimate implies the following:

\begin{lemma}\label{minimalVelocityWithA}
For all $\psi$ such that the right-hand side makes sense, and $t>1$, we have

\begin{equation}\label{minimalVelocityWithA2}
\Vert F(\frac{A}{t}<b) e^{-iHt} E_\bigtriangleup(H) \psi \| \lesssim t^{-5/4} (\| \psi \|^2 + \| |A|^{5/4} \psi \|^2)^\frac{1}{2}
\end{equation}

where $\bigtriangleup$ is any interval in the continuous spectrum of $H$, and $b$ is any constant less than $\theta$ (the constant appearing in the Mourre estimate for that interval). The exponent $\frac{5}{4}$ is not optimal.

\end{lemma}
\begin{proof}
The estimate (\ref{minimalVelocityWithA2}) follows immediately from the Mourre estimate and the abstract theory in \cite{ss2}. The constant $-5/4$ is not the best attainable, but it's sufficient for our purposes.
\end{proof}

The goal is to swap out the reference to the auxiliary operator $A$ with a reference to $x$. To this end, we will prove:

\begin{lemma}\label{minimalVelocityLemma}
For all $\psi$ such that the right-hand side makes sense, and $t>1$, we have

\begin{equation}\label{minimalVelocity}
\lim_{t\rightarrow \pm \infty} F(\frac{X^2}{t^{2-\epsilon}}<\delta) e^{-iHt} E_\bigtriangleup(H) \psi = 0 \tag{MV}
\end{equation}

where $\epsilon$ is any small positive constant, $\bigtriangleup$ is any interval in the continuous spectrum of $H$, and  $\delta$ is any positive constant less than $\theta$ (the constant appearing in the Mourre estimate for $\bigtriangleup$).
\end{lemma}

We need to take a few steps before we can prove this. The operator $F(\frac{X^2}{t^{2-\epsilon}}<\delta)$ can be written as $$ F(\frac{A}{t}<b) F(\frac{X^2}{t^{2-\epsilon}}<\delta) + F(\frac{A}{t}\ge b) F(\frac{X^2}{t^{2-\epsilon}}<\delta)$$

Since $$ \|F(\frac{A}{t}<b) F(\frac{X^2}{t^{2-\epsilon}}<\delta) e^{-iHt} E_\bigtriangleup(H) \psi\|  \lesssim t^{-5/4} (\| \psi \|^2 + \| |A|^{5/4} \psi \|^2)^\frac{1}{2}$$

by (\ref{minimalVelocityWithA2}), we need only prove that
$$ \|F(\frac{A}{t}\ge b) F(\frac{X^2}{t^{2-\epsilon}}<\delta) e^{-iHt} E_\bigtriangleup(H) \psi\|  \lesssim t^{-1} \|\psi \|$$
For expedience of notation, define the following:

$$F_1 = F(\frac{X^2}{t^{2-\epsilon}}<\delta)$$
$$\widetilde{F_1} = F(\frac{X^2}{t^{2-\epsilon}}<\frac{\delta}{2})$$
Note that $F_1 \widetilde{F_1} = F_1$ if the cutoffs are made sufficiently sharp; this is the point of using $\frac{\delta}{2}$.
$$ F_2(A) := F(\frac{A}{t}\ge b)$$
$$g := g(H)$$
where $g$ is a smoothed version of the energy cutoff function $E_\bigtriangleup$ so that $g(H) E_\bigtriangleup(H)= E_\bigtriangleup(H)$. Let $\widetilde{g}$ be a less sharp version of the smoothed cutoff function $g$ so that $g \widetilde{g} = g$.
$$\widetilde{A}= \widetilde{F_1} \widetilde{g} A \widetilde{g} \widetilde{F_1}$$
$$F_2(\widetilde{A}) = F(\frac{\widetilde{A}}{t}\ge b)$$

Thus, the thing to be estimated is: $$\|F_2(A) F_1 e^{-iHt} E_\bigtriangleup(H) \psi\|$$ Ensuing computations are much simplified by understanding some commutators of these operators. Here, $s$ is a constant, and $O(t^n)$ represents an operator with norm bounded by a constant times $t^n$.

$$[\frac{A}{t},F_1]  \approx \frac{X^2}{t^{3-\epsilon}} F'(\frac{X^2}{t^{2-\epsilon}}<\delta) = O(t^{-1}) $$
Then, (by e.g. Lemma 4.12 in \cite{cfks}) $[A,g]$ is bounded, so $$[\frac{A}{t},g]=O(t^{-1})$$

The next objective is to show that $[F_1,g]\rightarrow 0$ as $t\rightarrow \infty$. In this analysis, we treat $F_1= F(\frac{X^2}{t^{2-\epsilon}}<\delta)$ as a function of the vector $X/(t^{1-\epsilon/2})$. We write $\eta=1-\frac{\epsilon}{2}$ for convenience.

$$[F_1,g]= \int_{\mathbb{R}^6_\lambda} \widehat{F_1}(\lambda) [e^{i\lambda\cdot X/t^\eta},g] d\lambda$$

$$=  \int_{\mathbb{R}^6_\lambda} \widehat{F_1}(\lambda) e^{i\lambda\cdot X/t^\eta}\Big( g - e^{-i\lambda\cdot X/t^\eta} g e^{i\lambda\cdot X/t^\eta}\Big) d\lambda$$

$$\approx \int_{\mathbb{R}^6_\lambda} \widehat{F_1}(\lambda) e^{i\lambda\cdot X/t^\eta} \int_0^1 e^{-i s \lambda \cdot X/t^\eta} [g,X/t^\eta] e^{is \lambda \cdot X/t^\eta}\cdot \lambda \text{ } ds \text{ } d\lambda$$

Above, we have parameterized the line segment in $\mathbb{R}^6$ connecting the origin to $\lambda$ by $ \lambda s $ for $0\le s \le 1$. All we'd need to do now is show that $[g,X]$ gives a vector of bounded operators. Then, we'd have

 $$[F_1,g]\approx \frac{1}{t^\eta} \int_{\mathbb{R}^6_\lambda} \widehat{F_1}(\lambda) \lambda\cdot  e^{i\lambda\cdot X/t^\eta} \int_0^1 e^{-i s \lambda \cdot X/t^\eta} [g,X] e^{is \lambda \cdot X/t^\eta} \text{ } ds \text{ } d\lambda$$
$$ = \frac{1}{t^\eta} \int_{\mathbb{R}^6_\lambda} \widehat{\bigtriangledown F_1}(\lambda) \cdot  e^{i\lambda\cdot X/t^\eta} \int_0^1 e^{-i s \lambda \cdot X/t^\eta} [g,X] e^{is \lambda \cdot X/t^\eta} \text{ } ds \text{ } d\lambda$$
and the integral above would be norm convergent to a bounded operator. Unfortunately, $[g,X]$ is unbounded. So we estimate $[g,X]$.

$$[g(H),X]\approx \int_{-\infty}^\infty \widehat{g}(\lambda) e^{i\lambda H} \int_0^\lambda e^{-isH} [H,X] e^{isH} \text{ } ds \text{ } d\lambda$$

Now, $[H,x] = p + \frac{k}{|k|}$. Since $\frac{k}{|k|}$ is bounded after all, we may focus our analysis on $p$. It remains to estimate:

 $$ \int_{-\infty}^\infty \widehat{g}(\lambda) e^{i\lambda H} \int_0^\lambda e^{-isH} p e^{isH} \text{ } ds \text{ } d\lambda$$

 It suffices to bound

  $$ \int_{-\infty}^\infty \widehat{g}(\lambda) \Big( \frac{d}{d\lambda} e^{i\lambda H} \Big)\frac{1}{H+i} \int_0^\lambda e^{-isH} p e^{isH} \text{ } ds \text{ } d\lambda$$

because the difference between this and the desired expression is bounded. Integrating by parts, we find that this equals

\begin{align*}
&\int_{-\infty}^\infty \widehat{g}'(\lambda)  e^{i\lambda H} \frac{1}{H+i} \int_0^\lambda e^{-isH} p e^{isH} \text{ } ds \text{ } d\lambda \\
&\hspace{1cm} + \int_{-\infty}^\infty \widehat{g}(\lambda) \frac{1}{H+i} p e^{i\lambda H}  \text{ } d\lambda\\
\end{align*}

which is certainly bounded. Thus we conclude that $ [F_1,g] = O(t^{-\eta})$. Continuing, we have
\begin{align*}
[F_1g,\frac{A}{t}] &= O(t^{-\eta})\\
\end{align*}
\begin{align*}
[F_1g,e^{is\frac{A}{t}}] &= e^{is\frac{A}{t}}\Big( F_1g - e^{-is\frac{A}{t}} F_1g e^{is\frac{A}{t}} \Big)\\
&=e^{is\frac{A}{t}} \int_0^s e^{-ir\frac{A}{t}} [F_1g, \frac{A}{t}] e^{ir\frac{A}{t}} dr\\
 &= e^{is\frac{A}{t}}sO(t^{-\eta})\\
\end{align*}

The same estimates all hold when $F_1$ and $g$ are replaced by $\widetilde{F_1}$ and $\widetilde{g}$. \bigskip

Our first objective towards the proof of (\ref{minimalVelocity}) will be the following.
\begin{lemma}
$\| \frac{\widetilde{A}}{t} \| \le \frac{b}{2} $ for all sufficiently large $t$.
\end{lemma}
\begin{proof}
Compute
\begin{align*}
\frac{\widetilde{A}}{t} &= \widetilde{F_1} \widetilde{g} \frac{A}{t} \widetilde{g}\widetilde{F_1} \\
&\approx \widetilde{F_1} \widetilde{g} \frac{P\cdot X}{t} \widetilde{g} \widetilde{F_1} +\widetilde{F_1} \widetilde{g} \frac{X\cdot P}{t} \widetilde{g} \widetilde{F_1} \\
&= \widetilde{F_1} \widetilde{g} \frac{P\cdot X}{t} [\widetilde{g}, \widetilde{F_1}] +[\widetilde{F_1}, \widetilde{g}] \frac{X\cdot P}{t} \widetilde{g} \widetilde{F_1}  + \widetilde{F_1} \widetilde{g} \frac{P\cdot X}{t} \widetilde{F_1} \widetilde{g} +\widetilde{g} \widetilde{F_1} \frac{X\cdot P}{t} \widetilde{g} \widetilde{F_1}\\
\end{align*}

Since $\widetilde{g} p$ is bounded and $\frac{X}{t} \widetilde{F_1} = O(t^{-\epsilon/2})$, the last two terms above are $O(t^{-\epsilon/2})$. As for the first two terms:

\begin{align*}
\widetilde{F_1} \widetilde{g} \frac{P\cdot X}{t} [\widetilde{g}, \widetilde{F_1}] +[\widetilde{F_1}, \widetilde{g}] \frac{X\cdot P}{t} \widetilde{g} \widetilde{F_1} &= \widetilde{F_1}  \frac{[\widetilde{g} P, X]}{t} [\widetilde{g}, \widetilde{F_1}] +[\widetilde{F_1}, \widetilde{g}] \frac{[X, P \widetilde{g}]}{t}  \widetilde{F_1}\\
&\hspace{1cm} + \widetilde{F_1} \frac{X}{t}\cdot \widetilde{g} P [\widetilde{g}, \widetilde{F_1}] +[\widetilde{F_1}, \widetilde{g}] P \widetilde{g} \cdot \frac{X}{t} \widetilde{F_1}\\
&= O(t^{-\epsilon/2}) O(t^{-\eta}) + O(t^{-\epsilon/2})O(t^{-\eta})\\
&= O(t^{-1})\\
\end{align*}
\end{proof}

The implication of this lemma is that $F_2(\widetilde{A})$ actually equals $0$ for sufficiently large $t$. So to estimate $\|F_2(A) F_1 e^{-iHt} E_\bigtriangleup(H) \psi\|$, it is sufficient to estimate $\|\Big( F_2(\widetilde{A})-F_2(A)\Big) F_1 e^{-iHt} E_\bigtriangleup(H) \psi\|$.\bigskip

Next, we prove the \textbf{localization lemma}.

\begin{lemma}\label{localizationLemma}

$\| F_2(A)F_1E_\bigtriangleup(H) \| \rightarrow 0 $ as $t \rightarrow \infty$.
\end{lemma}

\begin{proof}
It's sufficient to prove that (notice it doesn't make a difference replacing $\widetilde{F_1}$ or $\widetilde{g}$ with its square or changing their order):
\begin{align*}
\lim_{t\rightarrow \infty} \| \Big( F_2(\widetilde{A})-F_2(A)\Big) (\widetilde{F_1}) (\widetilde{g})^2(\widetilde{F_1})  E_\bigtriangleup(H) e^{-iHt}\psi \| = 0 \\
\end{align*}

We (formally) express $F_2$ using a Fourier transform:
\begin{align*}
&\Big( F_2(\widetilde{A})-F_2(A)\Big)  (\widetilde{F_1}) (\widetilde{g})^2(\widetilde{F_1})  \\
&\approx \int \widehat{F_2}(\lambda) \Big( e^{-i\lambda\widetilde{A}/t} - e^{-i\lambda A/t} \Big) \hspace{.1cm}d\lambda \hspace{.1cm} (\widetilde{F_1}) (\widetilde{g})^2(\widetilde{F_1}) \\
&= \int \widehat{F_2}(\lambda) e^{-i\lambda\widetilde{A}/t} \int_0^\lambda  e^{i\lambda\widetilde{A}/t} \Big( \frac{\widetilde{A}}{t} - \frac{A}{t} \Big) e^{-i\lambda A/t} \hspace{.1cm}ds \hspace{.1cm}d\lambda \hspace{.1cm} (\widetilde{F_1}) (\widetilde{g})^2(\widetilde{F_1})  \\
&= \int \widehat{F_2}(\lambda) e^{-i\lambda\widetilde{A}/t} \int_0^\lambda  e^{i\lambda\widetilde{A}/t} \Big( \frac{\widetilde{A}}{t} - \frac{A}{t} \Big) [e^{-i\lambda A/t},(\widetilde{F_1}) (\widetilde{g})^2(\widetilde{F_1})] \hspace{.1cm}ds \hspace{.1cm}d\lambda \hspace{.1cm}  \\
&\hspace{1cm}+ \int \widehat{F_2}(\lambda) e^{-i\lambda\widetilde{A}/t} \int_0^\lambda  e^{i\lambda\widetilde{A}/t} \Big( \frac{\widetilde{A}}{t} - \frac{A}{t} \Big) (\widetilde{F_1}) (\widetilde{g})^2(\widetilde{F_1}) e^{-i\lambda A/t} \hspace{.1cm}ds \hspace{.1cm}d\lambda \hspace{.1cm}   \\
&= \int \widehat{F_2}(\lambda) e^{-i\lambda\widetilde{A}/t} \int_0^\lambda  e^{i\lambda\widetilde{A}/t} \Big( \frac{\widetilde{A}}{t} - \frac{A}{t} \Big) e^{-i\lambda A/t} sO(t^{-\eta}) \hspace{.1cm}ds \hspace{.1cm}d\lambda \hspace{.1cm}  \\
&\hspace{1cm}+ \int \widehat{F_2}(\lambda) e^{-i\lambda\widetilde{A}/t} \int_0^\lambda  e^{i\lambda\widetilde{A}/t} \Big( \frac{\widetilde{A}}{t} - \frac{A}{t} \Big) (\widetilde{F_1}) (\widetilde{g})^2(\widetilde{F_1}) e^{-i\lambda A/t} \hspace{.1cm}ds \hspace{.1cm}d\lambda \hspace{.1cm}   \\
\end{align*}

The first integral converges to an operator that is $O(t^{-\eta})$ (applying integration by parts once to the inner integral). The second integral is shown to also converges to an operator that is $O(t^{-\eta})$ by the following argument: we have

\begin{align*}
\Big( \frac{\widetilde{A}}{t} - \frac{A}{t} \Big) (\widetilde{F_1}) (\widetilde{g})^2(\widetilde{F_1}) &= \Big( F_1 g \frac{A}{t}g F_1 - \frac{A}{t} \Big) (\widetilde{F_1}) (\widetilde{g})^2(\widetilde{F_1}) \\
\end{align*}

If we commute one copy of $\widetilde{F_1} \widetilde{g}$ all the way to the left, then we are left with $(\widetilde{F_1}) (\widetilde{g})\frac{A}{t} (\widetilde{g})(\widetilde{F_1})- (\widetilde{F_1}) (\widetilde{g})\frac{A}{t} (\widetilde{g})(\widetilde{F_1}) = 0$, but using the fact that $F_1 g \frac{A}{t}$ is $O(t^{-\epsilon})$, each term of the commutator is at least $O(t^{-\epsilon})$.
\end{proof}

Now, (\ref{minimalVelocity}) is an immediate corollary of the localization lemma \ref{localizationLemma}.

\section{Asymptotic completeness}
\subsection{Existence of the wave operators}

In this section our aim is to prove existence of a collection of Deift-Simon wave operators arising from operators $\lbrace F_a : \#(a)=2 \rbrace$ that form a partition of unity. Let $\bigtriangleup$ be an interval in the continuous spectrum of $H$ such that $\bigtriangleup<0$, let $\delta>0$ (where $\delta$ is selected based on $\bigtriangleup$ to make the minimal velocity estimates hold), and let $\epsilon>0$ be a very small positive constant so that $(\delta-\epsilon)>0$. We write $\delta' := \delta -\epsilon$. We define

$$F_{a} = F(\frac{x_a^2}{t^{2-\epsilon}}<\delta')$$

Then, we have the following.

\begin{lemma}\label{waveOpsExist}
Let $\psi$ be a wavefunction such that $(\| \psi \|^2 + \| |A|^{5/4} \psi \|^2)^\frac{1}{2}<\infty$. The following limits exist for every $a$ such that $\#(a)=2$:
\begin{equation}
\lim_{t\rightarrow \pm \infty} E_{\bigtriangleup}(H_a)e^{iH_a t} F_a e^{-iHt}E_{\bigtriangleup}(H)\psi
\end{equation}
\end{lemma}

The necessary facts for the proof of Lemma \ref{waveOpsExist} are the minimal velocity estimate (\ref{minimalVelocity})), the \textbf{short range assumptions} for $\#(a)=2$:

\begin{equation}\label{SR}
\| F(\frac{x_a^2}{t^{2-\epsilon}}>c)I^a \| \text{ converges to 0 and is integrable in } t\tag{SR}
\end{equation}

where $c>0$ is any constant, and \textbf{fast decay of the eigenfunctions} for $\#(a)=2$:

\begin{equation}\label{FDE}
E_\bigtriangleup(H_a) \langle x_a \rangle^{4} \text{ is a bounded operator.}\tag{FDE}
\end{equation}

where $\bigtriangleup$ is an interval below $0$. The fast decay of the eigenfunctions is where we are using the negativity of $\bigtriangleup$; the above energy cutoffs are then projections onto eigenvalues of subsystems, so $(\ref{FDE})$ merely alleges that these eigenfunctions decay rapidly.\bigskip

The existence of the limits in Lemma \ref{waveOpsExist} is connected to time estimates on Heisenberg derivatives.

\begin{lemma}\label{heisDerivsAreIntegrable}
Let $\psi$ be a wavefunction such that $(\| \psi \|^2 + \| |A|^{5/4} \psi \|^2)^\frac{1}{2}<\infty$. The following are integrable functions of $t$ near $t=\infty$ for $\#(a)=2$:
\begin{equation}\label{heisDerivs1}
\Vert E_{\bigtriangleup}(H_a)e^{iH_a t} \Bigg( I_a F_a + \frac{d}{dt} F_a +[H_0,F_a]\Bigg) e^{-iHt}E_{\bigtriangleup}(H)\psi\Vert
\end{equation}
\end{lemma}

First, we note that Lemma \ref{heisDerivsAreIntegrable} implies Lemma \ref{waveOpsExist}. \bigskip

We want to show that $\Vert E_{\bigtriangleup}(H_a)e^{iH_a t_1} F_a e^{-iHt_1}E_{\bigtriangleup}(H)\psi - E_{\bigtriangleup}(H_a)e^{iH_a t_2} F_a e^{-iHt_2}E_{\bigtriangleup}(H)\psi\Vert$ can be made arbitrarily small by taking $t_1$ and $t_2$ large enough. We can rewrite:

\begin{align*}
&\Vert E_{\bigtriangleup}(H_a)e^{iH_a t_1} F_a e^{-iHt_1}E_{\bigtriangleup}(H)\psi - E_{\bigtriangleup}(H_a)e^{iH_a t_2} F_a e^{-iHt_2}E_{\bigtriangleup}(H)\psi\Vert \\
&= \sup_{\Vert \phi \Vert = 1} \Big| \langle E_{\bigtriangleup } (H_a) \phi , e^{iH_a t_1} F_a e^{-iHt_1}E_{\bigtriangleup}(H)\psi \rangle - \langle E_{\bigtriangleup  } (H_a) \phi,e^{iH_a t_2} F_a e^{-iHt_2}E_{\bigtriangleup}(H)\psi\rangle \Big| \\
&\le \sup_{\Vert \phi \Vert = 1} \int_{t_1}^{t_2} \Big|\frac{d}{dt} \langle E_{\bigtriangleup } (H_a)\phi, e^{iH_a t} F_a e^{-iHt}E_{\bigtriangleup}(H)\psi \rangle  \Big| dt \\
&\le \int_{t_1}^{t_2}\Vert E_{\bigtriangleup}(H_a)e^{iH_a t} \Bigg( I_a F_a + \frac{d}{dt} F_a +[H_0,F_a]\Bigg) e^{-iHt}E_{\bigtriangleup}(H)\psi\Vert dt\\
\end{align*}

We move on prove that (\ref{heisDerivs1}) is an integrable function of $t$ for $a=(y)(x0)$, the simplest case. We estimate term by term. Consider first the term containing $$I_a F_a= \left( V_{13}(y) + V_{23}(x-y)\right) F(\frac{x^2}{t^{2-\epsilon}}<\delta') $$

Since we can partition unity into $1 = F(\frac{y^2}{t^{2-\epsilon}}<\epsilon) + F(\frac{y^2}{t^{2-\epsilon}}>\epsilon)$, we have

$$I_a F_a = \left( V_{13}(y) + V_{23}(x-y)\right) \left( F(\frac{x^2}{t^{2-\epsilon}}<\delta')F(\frac{y^2}{t^{2-\epsilon}}<\epsilon) + F(\frac{x^2}{t^{2-\epsilon}}<\delta')F(\frac{y^2}{t^{2-\epsilon}}>\epsilon)\right) $$

\begin{lemma}\label{usingEverything}
The term

$$\Vert E_{\bigtriangleup}(H_a) \Bigg(\left( V_{13}(y) + V_{23}(x-y)\right) \left( F(\frac{x^2}{t^{2-\epsilon}}<\delta')F(\frac{y^2}{t^{2-\epsilon}}<\epsilon)\right) \Bigg) e^{-iHt}E_{\bigtriangleup}(H)\psi\Vert$$

 is integrable in $t$.
\end{lemma}

The proof of Lemma (\ref{usingEverything}) is as follows. By fast decay of the eigenfunctions, it suffices to prove that

$$\Vert <x>^{-4} \Bigg(\left( V_{13}(y) + V_{23}(x-y)\right) \left( F(\frac{x^2}{t^{2-\epsilon}}<\delta')F(\frac{y^2}{t^{2-\epsilon}}<\epsilon)\right) \Bigg) e^{-iHt}E_{\bigtriangleup}(H)\psi\Vert$$

 is integrable in $t$. But then by the short range assumption, it is sufficient to prove that

 $$\Vert <X>^{-(1+\epsilon)} \left( F(\frac{x^2}{t^{2-\epsilon}}<\delta')F(\frac{y^2}{t^{2-\epsilon}}<\epsilon)\right) \Bigg) e^{-iHt}E_{\bigtriangleup}(H)\psi\Vert$$

is integrable in $t$. But this is true if

$$\Vert <X>^{-(1+\epsilon)} F(\frac{X^2}{t^{2-\epsilon}}<\delta)  e^{-iHt}E_{\bigtriangleup}(H)\psi\Vert$$

is integrable in $t$. So it is sufficient to see if

$$ \Vert <X>^{-(1+\epsilon)} F(\frac{A}{t}<b) E_{\bigtriangleup}(H)e^{-iHt} \psi \Vert$$

$$ \Vert <X>^{-(1+\epsilon)}  F(\frac{A}{t}> b) E_{\bigtriangleup}(H)e^{-iHt} \psi \Vert$$

are both integrable in $t$. Evidently the former is integrable in $t$, from Lemma \ref{minimalVelocityWithA}. The latter is proven integrable in $t$ as follows. \bigskip

We can prove $$ \Vert <X>^{-(1+\epsilon)}  F(\frac{A}{t}>b) E_{\bigtriangleup}(H)\Vert$$

is $O(t^{-\alpha})$ for $\alpha=1,2$ and then use complex interpolation to conclude. We write

\begin{align*}
\Vert <X>^{-\alpha}  F(\frac{A}{t}>b) E_{\bigtriangleup}(H)\Vert \lesssim \frac{1}{t^\alpha} \Vert <X>^{-\alpha}  F(\frac{A}{t}>b)  A^\alpha E_{\bigtriangleup}(H)\Vert
\end{align*}

Since the difference between $\tanh(b-\frac{A}{t})+1$ and $F(\frac{A}{t}>b)$ decays fast at infinity and therefore $$\big( \tanh(b-\frac{A}{t})+1 - F(\frac{A}{t}>b)\big) A^\alpha$$ is bounded, it is sufficient to prove that

$$\Vert <X>^{-\alpha}  \big( \tanh(b-\frac{A}{t})+1\big)  A^\alpha E_{\bigtriangleup}(H)\Vert$$

is bounded. We let $F_3=  F_3(A) = \big( \tanh(b-\frac{A}{t})+1\big)$. Then since $<P>^\alpha E_\bigtriangleup(H)$ is bounded, it is sufficient to prove

$$\Vert <X>^{-\alpha}  F_3  <X>^\alpha \Vert$$

is a bounded operator. In what follows we use the fact that $F_3$ is analytic in a strip of width greater than 1 and containing the real line.

\begin{align*}
 <X>^{-\alpha}  F_3  <X>^\alpha &= (bounded) + [<X>^{-\alpha},  F_3 ]  <X>^\alpha\\
 &= <X>^{-\alpha} [<X>^{\alpha},  F_3] \\
 &\approx <X>^{-\alpha} \int \int_0^\lambda e^{i s A/t} [<X>^{\alpha},  \frac{A}{t}] e^{-i s A/t} \hspace{.1cm} ds \hspace{.1cm} e^{i \lambda A/t} \widehat{F_3}(\lambda) \hspace{.1cm} d\lambda \\
 &\approx <X>^{-\alpha} \int \int_0^\lambda e^{i s A/t} \frac{<X>^{\alpha}}{t} e^{-i s A/t} \hspace{.1cm} ds \hspace{.1cm} e^{i \lambda A/t} \widehat{F_3}(\lambda) \hspace{.1cm} d\lambda \\
 &\approx <X>^{-\alpha} \int \int_0^\lambda e^{i s/t} \frac{<X>^{\alpha}}{t}  \hspace{.1cm} ds \hspace{.1cm} e^{i \lambda A/t} \widehat{F_3}(\lambda) \hspace{.1cm} d\lambda \\
 &\approx <X>^{-\alpha} \int (e^{i \lambda/t}-1) <X>^{\alpha}   \hspace{.1cm} e^{i \lambda A/t} \widehat{F_3}(\lambda) \hspace{.1cm} d\lambda \\
  &\approx \int (e^{i \lambda/t}-1) \hspace{.1cm} e^{i \lambda A/t} \widehat{F_3}(\lambda) \hspace{.1cm} d\lambda \\
  &\approx F_3(A+i) - F_3(A)\\
\end{align*}

This is bounded, so using Stein's interpolation theorem, we get the desired result. Moving on,

 $$\Vert E_{\bigtriangleup}(H_a) \Bigg(\left( V_{13}(y) + V_{23}(x-y)\right) \left( F(\frac{x^2}{t^{2-\epsilon}}<\delta')F(\frac{y^2}{t^{2-\epsilon}}>\epsilon)\right) \Bigg) e^{-iHt}E_{\bigtriangleup}(H)\psi\Vert$$

 is integrable in $t$ because of the short range assumption (\ref{SR}). \bigskip

Second, consider the term containing $\frac{d}{dt} F(\frac{x^2}{t^{2-\epsilon}}<\delta')$, which is a constant times $\frac{x^2}{t^{3-\epsilon}}F'(\frac{x^2}{t^{2-\epsilon}}<\delta')$. Since from (\ref{FDE}) $E_\bigtriangleup(H_a)$ is a bounded operator times $\langle x \rangle^{-4}$, the term

 $$\Vert E_{\bigtriangleup}(H_a) \Bigg( \frac{x^2}{t^{3-\epsilon}}F'(\frac{x^2}{t^{2-\epsilon}}<\delta') \Bigg) e^{-iHt}E_{\bigtriangleup}(H)\psi\Vert$$

decays in $t$ at least as fast as $t^{-(5-2\epsilon)}$ (and so is integrable). This is because $F'(\frac{x^2}{t^{2-\epsilon}}<\delta')$ is supported only on the configuration space region where $x^2\approx {t^{2-\epsilon}}$; we will use this fact repeatedly.\bigskip

Third, consider $[H_0,F(\frac{x^2}{t^{2-\epsilon}}<\delta')]$, which can be rewritten as $[p^2,F(\frac{x^2}{t^{2-\epsilon}}<\delta')]$ and is therefore equal to $$-2i p \cdot F'(\frac{x^2}{t^{2-\epsilon}}<\delta')\frac{2x}{t^{2-\epsilon}} -\bigtriangleup (F(\frac{x^2}{t^{2-\epsilon}}<\delta'))$$

 $$= -2i p \cdot F'(\frac{x^2}{t^{2-\epsilon}}<\delta')\frac{2x}{t^{2-\epsilon}} -F''(\frac{x^2}{t^{2-\epsilon}}<\delta'))\frac{4x^2}{t^{2(2-\epsilon)}} - F'(\frac{x^2}{t^{2-\epsilon}}<\delta')\frac{2}{t^{2-\epsilon}}$$

This in mind, the following are integrable in $t$ by (\ref{FDE}).

$$\Vert E_{\bigtriangleup}(H_a)\Bigg(p \cdot F'(\frac{x^2}{t^{2-\epsilon}}<\delta')\frac{2x}{t^{2-\epsilon}} \Bigg) e^{-iHt}E_{\bigtriangleup}(H)\psi\Vert$$
$$\Vert E_{\bigtriangleup}(H_a) \Bigg(F''(\frac{x^2}{t^{2-\epsilon}}<\delta'))\frac{4x^2}{t^{2(2-\epsilon)}} \Bigg) e^{-iHt}E_{\bigtriangleup}(H)\psi\Vert$$
$$\Vert E_{\bigtriangleup}(H_a) \Bigg(F'(\frac{x^2}{t^{2-\epsilon}}<\delta')\frac{2}{t^{2-\epsilon}} \Bigg) e^{-iHt}E_{\bigtriangleup}(H)\psi\Vert$$

Therefore so is the whole term

$$\Vert E_{\bigtriangleup}(H_a) \Bigg([H_0,F(\frac{x^2}{t^{2-\epsilon}}<\delta')] \Bigg) e^{-iHt}E_{\bigtriangleup}(H)\psi\Vert$$

and the proof is finished for the case $a=(y)(x0)$. We move on to the case $a=(x)(y0)$. Again we estimate term by term. The estimate of the first two terms (those containing $I_a F_a$ and $\frac{d}{dt}F_a$) are so similar to the previous case that we omit them. We investigate in detail the term containing $[H_0,F(\frac{y^2}{t^{2-\epsilon}}<\delta')]$, which can be rewritten as $[|k|,F(\frac{x^2}{t^{2-\epsilon}}<\delta']$. The approach here is to use the square root expression for $|k|$. We build up to this with a series of lemmas.

\begin{lemma}\label{917first}
The following operator-valued integral converges in norm to a bounded operator.
$$\int_0^\infty \langle y \rangle^{-4} \frac{s^{-1/2}}{s+k^2}ds$$
\end{lemma}
\begin{proof}
We have that $\langle y \rangle^{-4}$ is a bounded operator, and $$ \| \frac{s^{-1/2}}{s+k^2} \|_{op} \approx s^{-3/2}$$, so the integrand has sufficient decay near $s=\infty$. To estimate the integrand near $s=0$, we use Hardy-Littlewood-Sobolev (where $6/5$ is a non-optimal choice):
$$ \langle y \rangle^{-4} \frac{s^{-1/2}}{s+k^2} = \langle y \rangle^{-4} \frac{1}{|k|^{6/5}} |k|^{6/5} \frac{s^{-1/2}}{s+k^2}$$
Since $\langle y \rangle^{-6/5} \frac{1}{|k|^{6/5}}$ is a bounded operator on $L^2(\mathbb{R}^3)$ and $\| |k|^{6/5} \frac{s^{-1/2}}{s+k^2} \|_{op} \lesssim s^{-9/10}$, we have our conclusion.
\end{proof}

\begin{lemma}\label{917second}
The following operator-valued integrals converge to bounded operators.
\begin{equation}\label{9175}
\int_0^\infty \langle y \rangle^{-4}  \frac{s^{-1/2}}{s+k^2} 2y \cdot q E_\bigtriangleup(H) ds
\end{equation}
\begin{equation}\label{9176_1}
\int_0^\infty \langle y \rangle^{-4} \frac{s^{-1/2}}{s+k^2} y^2 E_\bigtriangleup(H) ds
\end{equation}
\begin{equation}\label{9176_2}
\int_0^\infty \langle y \rangle^{-4} \frac{s^{-1/2}}{s+k^2} E_\bigtriangleup(H) ds
\end{equation}
\begin{equation}\label{9177}
\int_0^\infty \langle y \rangle^{-4}  \frac{s^{-1/2}}{s+k^2}2y \cdot k \frac{k^2}{s+k^2}E_\bigtriangleup(H) ds
\end{equation}
\begin{equation}\label{9178_1}
\int_0^\infty \langle y \rangle^{-4} \frac{s^{-1/2}}{s+k^2}\left( y^2\right) \frac{k^2}{s+k^2}E_\bigtriangleup(H) ds
\end{equation}
\begin{equation}\label{9178_2}
\int_0^\infty \langle y \rangle^{-4} \frac{s^{-1/2}}{s+k^2} \frac{k^2}{s+k^2}E_\bigtriangleup(H) ds
\end{equation}
\end{lemma}
\begin{proof}
In all cases, commute all $y$ all the way to the left and use Hardy-Littlewood-Sobolev if necessary, as in the proof of Lemma \ref{917first}.
\end{proof}

\begin{lemma}\label{917third}
The following operator-valued integrals converge to bounded operators, which are integrable in $t$ near $t=\infty$.
\begin{equation}\label{9171}
\int_0^\infty \langle y \rangle^{-4}  \frac{s^{-1/2}}{s+k^2}F'(\frac{y^2}{t^{2-\epsilon}}<\delta') \frac{2y}{t^{2-\epsilon}}\cdot k E_\bigtriangleup(H) ds
\end{equation}
\begin{equation}\label{9172}
\int_0^\infty \langle y \rangle^{-4} \frac{s^{-1/2}}{s+k^2}\left( F''(\frac{y^2}{t^{2-\epsilon}}<\delta') \frac{4y^2}{t^{2(2-\epsilon)}} + F'(\frac{y^2}{t^{2-\epsilon}}<\delta') \frac{2}{t^{2-\epsilon}}\right) E_\bigtriangleup(H) ds
\end{equation}
\begin{equation}\label{9173}
\int_0^\infty \langle y \rangle^{-4}  \frac{s^{-1/2}}{s+k^2}F'(\frac{y^2}{t^{2-\epsilon}}<\delta') \frac{2y}{t^{2-\epsilon}}\cdot k \frac{k^2}{s+k^2}E_\bigtriangleup(H) ds
\end{equation}
\begin{equation}\label{9174}
\int_0^\infty \langle y \rangle^{-4} \frac{s^{-1/2}}{s+k^2}\left( F''(\frac{y^2}{t^{2-\epsilon}}<\delta') \frac{4y^2}{t^{2(2-\epsilon)}} + F'(\frac{y^2}{t^{2-\epsilon}}<\delta') \frac{2}{t^{2-\epsilon}}\right) \frac{k^2}{s+k^2}E_\bigtriangleup(H) ds
\end{equation}
\end{lemma}
\begin{proof}
Applying Lemma \ref{917second} and the fact that $F'(\frac{y^2}{t^{2-\epsilon}}<\delta')\frac{1}{t^{2-\epsilon}}$ and $F''(\frac{y^2}{t^{2-\epsilon}}<\delta')\frac{1}{t^{2(2-\epsilon)}}$ are bounded operators that are integrable in $t$ near $t=\infty$, this is immediate.
\end{proof}

\begin{lemma}\label{917fourth}
The following operator-valued integral converges to a bounded operator that is integrable in $t$ near $t=\infty$.
$$\int_0^\infty \langle y \rangle^{-4} [\frac{s^{-1/2}}{s+k^2}k^2, F(\frac{y^2}{t^{2-\epsilon}}<\delta')] E_\bigtriangleup(H) ds$$
\end{lemma}
\begin{proof}
Expanding this commutator, we get several terms as in (\ref{9171})-(\ref{9174}).
\end{proof}

Thus, we claim the term

$$\| E_\bigtriangleup(H_a)e^{iH_a t} \Bigg( [|k|,F(\frac{y^2}{t^{2-\epsilon}}<\delta')] \Bigg) e^{-iHt} E_\bigtriangleup(H)\| $$ is integrable in $t$; since $E_{\bigtriangleup}(H_a)$ is a bounded operator times $\langle y \rangle^{-4}$ and we can use the square-root representation of $|k|$, this reduces to Lemma \ref{917fourth}. This concludes the proof for the case $a=(x)(y0)$.\bigskip

Finally, we mention that a substantially similar proof gives that (\ref{heisDerivs1}) is integrable in $t$ for the case $a=(xy)(0)$. Proceeding as in the previous case, simply replace references to $x$ or $y$ with references to $(x-y)$ or $(x+y)$. \bigskip

This proves Lemma \ref{waveOpsExist}. The next goal is to show that the energy cutoffs on the left hand side are actually redundant.

\begin{lemma}\label{noEnergyCutoffs}
Let $\psi$ be a wavefunction such that that $ \| \psi \| + \| |A|^{5/4} \psi \|<\infty$. The following limits exist for $\#(a)=2$:
\begin{equation}
\lim_{t\rightarrow \pm \infty} e^{iH_a t} F_a e^{-iHt}E_{\bigtriangleup}(H)\psi
\end{equation}
\end{lemma}

Once this is established, then because such $\psi$ form a dense set, it is immediate that the strong limits $$s-\lim_{t\rightarrow \pm \infty} e^{iH_a t} F_a e^{-iHt}E_{\bigtriangleup}(H)$$ exist. This is the fact that we will use in the proof of asymptotic completeness.

\begin{proof}
 By Lemma \ref{waveOpsExist}, we need to show that $\lim_{t\rightarrow \pm \infty} (1- E_\bigtriangleup(H_a))e^{iH_a t} F_a e^{-iHt}E_{\bigtriangleup}(H)\psi$ exists in all three cases $a$. We can rewrite
\begin{align*}
(1- E_\bigtriangleup(H_a))e^{iH_a t} F_a e^{-iHt}E_{\bigtriangleup}(H)\psi &= e^{iH_a t} [- E_\bigtriangleup(H_a),F_a] e^{-iHt}E_{\bigtriangleup}(H)\psi  \\
&\hspace{1cm} + e^{iH_a t} F_a (E_{\bigtriangleup}(H)\psi- E_\bigtriangleup(H_a))e^{-iHt}  E_{\bigtriangleup}(H)\psi\\
\end{align*}

By the usual Stone-Weierstrass argument (replacing $E_\bigtriangleup(H_a)$ with a smoothed version), this converges to zero if the following converge to $0$.

\begin{equation}\label{9251}
e^{iH_a t} [- \frac{1}{H_a+i},F_a] e^{-iHt}E_{\bigtriangleup}(H) \psi
\end{equation}
\begin{equation}\label{9252}
e^{iH_a t} F_a (\frac{1}{H+i}- \frac{1}{H_a+1})e^{-iHt}  E_{\bigtriangleup}(H)\psi
\end{equation}

Regarding (\ref{9251}), we compute
\begin{align*}
[\frac{1}{H_a+i},F_a]  &= \frac{1}{H_a +i}[H_0,F_a] \frac{1}{H_a+i}\\
\end{align*}

By estimating $[p^2,F_a]$ and $[|k|,F_a]$ as before, one proves that (\ref{9251}) converges to $0$. This estimate is strictly easier to prove than those in Lemma \ref{waveOpsExist} because we merely to show convergence to $0$, not integrability in $t$. \bigskip

As for (\ref{9252}), we compute

\begin{align*}
F_a (\frac{1}{H+i}- \frac{1}{H_a+1}) &= [F_a,\frac{1}{H+i}] I_a \frac{1}{H_a+i} + \frac{1}{H+i}F_a I_a \frac{1}{H_a+i}\\
\end{align*}
The analysis of the term containing $[F_a,\frac{1}{H+i}] I_a \frac{1}{H_a+i}$ reduces to the same analysis as for (\ref{9251}), so we need only consider the term containing $\frac{1}{H+i}F_a I_a \frac{1}{H_a+i}$. As an example we show how this works for the case $a=(y)(x0)$.

$$\frac{1}{H+i}\Big(V_{13}(y) + V_{23}(x-y)\Big) F(\frac{x^2}{t^{2-\epsilon}}<\delta') \frac{1}{H_a+i} $$ $$= \frac{1}{H+i}\Big(V_{13}(y) + V_{23}(x-y)\Big) F(\frac{x^2}{t^{2-\epsilon}}<\delta')F(\frac{y^2}{t^{2-\epsilon}}>\epsilon)\frac{1}{H_a+i} $$
$$+ \frac{1}{H+i}\Big(V_{13}(y) + V_{23}(x-y)\Big) F(\frac{x^2}{t^{2-\epsilon}}<\delta')F(\frac{y^2}{t^{2-\epsilon}}<\epsilon)\frac{1}{H_a+i} $$

The first term here can be taken care of with the short range property (\ref{SR}), i.e.

$$\| F_a \frac{1}{H+i}I_a  F(\frac{x^2}{t^{2-\epsilon}}<\delta')F(\frac{y^2}{t^{2-\epsilon}}>\epsilon) \frac{1}{H_a+1}e^{-iHt}  E_{\bigtriangleup}(H)\psi \|$$

converges to $0$ in $t$. For the second term, one wishes to use the minimal velocity estimate (\ref{minimalVelocity}). We commute

$$\frac{1}{H+i}\Big(V_{13}(y) + V_{23}(x-y)\Big) F(\frac{x^2}{t^{2-\epsilon}}<\delta')F(\frac{y^2}{t^{2-\epsilon}}<\epsilon)\frac{1}{H_a+i}$$
$$= \frac{1}{H+i}\Big(V_{13}(y) + V_{23}(x-y)\Big) [F(\frac{x^2}{t^{2-\epsilon}}<\delta')F(\frac{y^2}{t^{2-\epsilon}}<\epsilon),\frac{1}{H_a+i}]$$
$$+ \frac{1}{H+i}\Big(V_{13}(y) + V_{23}(x-y)\Big))\frac{1}{H_a+i}F(\frac{x^2}{t^{2-\epsilon}}<\delta')F(\frac{y^2}{t^{2-\epsilon}}<\epsilon)$$
The first term above once again reduces to the same analysis as (\ref{9251}). Then we can use the minimal velocity estimate on the second term, i.e.

$$\| F_a \frac{1}{H+i}I_a \frac{1}{H_a+1} F(\frac{x^2}{t^{2-\epsilon}}<\delta')F(\frac{y^2}{t^{2-\epsilon}}<\epsilon)e^{-iHt}  E_{\bigtriangleup}(H)\psi \|$$

converges to $0$ in $t$. Having shown all terms converge to $0$ in $t$, this concludes the proof.
\end{proof}

\subsection{Proof of the theorem}
Define the following operators $W_{a}(t)$ for $\#(a)=2$ by

\begin{equation}
W_{a}(t):= e^{iH_a t} F_a e^{-iHt}E_{\bigtriangleup}(H)
\end{equation}

Given a state $\psi$, we define $\phi_a:= \lim_{t\rightarrow \pm \infty} W_{a}(t) \psi$. \bigskip

We are ready to prove Theorem 2. Let $\psi$ be a state on the range of $E_\bigtriangleup(H)$. Following \cite{ss}, we write:

\begin{align*}
e^{-iHt} \psi &= \sum_{\#(a)=2}  F_a e^{-iHt} \psi + rem.\\
&= \sum_{\#(a)=2} e^{-iH_a t} W_a(t) \psi + rem.\\
\end{align*}

where $rem.$ converges to $0$ in $t$. Taking limits, we arrive at the statement of the theorem. It remains to show that the remainder does indeed decay.

\begin{align*}
rem. = (1- F(\frac{x^2}{t^{2-\epsilon}}<\delta')-F(\frac{y^2}{t^{2-\epsilon}}<\delta')- F(\frac{(x-y)^2}{t^{2-\epsilon}}<\delta')) e^{-iHt}E_{\bigtriangleup}(H)\psi\\
\end{align*}

By minimal velocity it is free to add:

\begin{align*}
rem. &= (1- F(\frac{x^2}{t^2}<\delta')F(\frac{y^2}{t^2}>\epsilon)-F(\frac{y^2}{t^2}<\delta')F(\frac{x^2}{t^2}>\epsilon)\\
&\hspace{1cm}- F(\frac{(x-y)^2}{t^2}<\delta')F(\frac{(x+y)^2}{t^2}>\epsilon)) e^{-iHt}E_{\bigtriangleup}(H)\psi\\
\end{align*}

since the rest converges to $0$ in $t$. Now, the operator $$(1- F(\frac{x^2}{t^2}<\delta')F(\frac{y^2}{t^2}>\epsilon)-F(\frac{y^2}{t^2}<\delta')F(\frac{x^2}{t^2}>\epsilon)- F(\frac{(x-y)^2}{t^2}<\epsilon)F(\frac{(x+y)^2}{t^2}>\epsilon)) $$ is supported in the phase space region where $x^2>\epsilon t^2$, $y^2 > \epsilon t^2$, and $(x-y)^2 > \epsilon t^2$. Because we are on negative energy, we may write $$E_\bigtriangleup(H)= E_\bigtriangleup(H)- E_\bigtriangleup(H_0)$$ and then use Stone-Weierstrass to obtain from this: $$\frac{1}{H_0+i}(V_{12}(x) + V_{13}(y) + V_{23}(x-y)) \frac{1}{H+i}$$ Commuting these phase space operators in, we see that the necessary decay is achieved. This concludes the proof.

\end{document}